\documentclass[preprint,12pt]{imsart}

\RequirePackage[OT1]{fontenc}
\RequirePackage{amsthm,amsmath}
\RequirePackage[numbers]{natbib}

\usepackage{amscd,amsfonts,amssymb,amsmath,latexsym,array,hhline,xcolor,graphicx}
\usepackage{float}
\usepackage{booktabs}
\usepackage{caption}
\newcommand\cA{{\cal A}}

\newcommand\cF{{\cal F}}

\newcommand\cB{{\cal B}}

\newcommand\cD{{\cal D}}
\newcommand\cR{{\cal R}}

\newcommand\cP{{\cal P}}
\newcommand\cQ{{\cal Q}}

\newcommand\cZ{{\cal Z}}

\newcommand\R{{\bf R}}

\newcommand\N{\mathbb{N}}

\newtheorem{theo}{Theorem}[section]
\newtheorem{prop}[theo]{Proposition}
\newtheorem{lemm}[theo]{Lemma}
\newtheorem{defi}[theo]{Definition}
\newtheorem{ex}[theo]{Example}

\newtheorem{coro}[theo]{Corollary}
\newtheorem{rem}[theo]{Remark}

\newcommand\beq{\begin{equation}}
\newcommand\eeq{\end{equation}}
\newcommand\beqa{\begin{equation*}}
\newcommand\eeqa{\end{equation*}}
\newcommand\bea{\begin{eqnarray}}
\newcommand\eea{\end{eqnarray}}
\newcommand\bean{\begin{eqnarray*}}
\newcommand\eean{\end{eqnarray*}}

\DeclareMathOperator{\essinf}{ess\;inf}
\DeclareMathOperator{\esssup}{ess\;sup}

\newcommand\bP{{\bf {\rm P}}}

\begin{document}

\begin{frontmatter}

\title{ Coherent Risk Measure on $L^0$:  NA Condition, Pricing and Dual Representation}

\author[A1]{ Emmanuel LEPINETTE,}
\author[A1]{ Duc Thinh VU}

\address[A1]{ Ceremade, UMR  CNRS 7534,  Paris Dauphine University, PSL National Research, Place du Mar\'echal De Lattre De Tassigny, \\
75775 Paris cedex 16, France and\\
Gosaef, Faculty of Sciences of Tunis, Tunisia.\\\smallskip
Email: emmanuel.lepinette@ceremade.dauphine.fr \quad   vu@ceremade.dauphine.fr
}

\begin{abstract}The NA condition is one of the pillars supporting the classical theory of financial mathematics. We revisit this condition for financial market models where a dynamic risk-measure defined on $L^0$ is fixed to characterize the family of acceptable wealths that play the role of non negative financial positions. We provide in this setting a new version of the fundamental theorem of asset pricing and we deduce a dual characterization of the super-hedging prices (called risk-hedging prices)  of a European option. Moreover, we show that  the set of all risk-hedging prices is closed under NA. At last, we provide  a dual representation of the risk-measure on $L^0$ under some conditions. \end{abstract}


\end{frontmatter}

\section{Introduction}
The NA condition originates from the work of Black and Scholes \cite{BS} and Merton \cite{Merton}.  In these articles, the risky asset is modeled by a geometric Brownian motion. The NA condition means  the absence of arbitrage opportunities, i.e. a nonzero terminal portfolio value  can not be acceptable if it starts from the zero initial endowment. A financial  position in the classical arbitrage theory is acceptable if it is  non negative almost surely.  In our work, the new contribution is that we consider a larger class of acceptable positions which are defined from a risk-measure. \smallskip

The NA condition is characterized through the famous Fundamental Theorem of Asset Pricing (FTAP) for a variety of financial models.  Essentially, NA is  equivalent to the existence of a so-called risk-neutral probability measure, under which the price process is a martingale. In discrete-time, the well known FTAP theorem  has been proved  by Dalang, Morton and Willinger \cite{DMW}. We may also mention the papers \cite{JS}, \cite{KK},  \cite{KSt}, \cite{Rogers},  \cite{Ross}. In continuous time, the formulation of the FTAP theorem is only possible once continuous-time self-financing portfolios are defined, see the seminal work of  Black and Scholes \cite{BS}. This gave rise to an extensive development of the stochastic calculus, e.g. for semi-martingales \cite{HK}, making possible  formulation of several versions of  the FTAP theorem as given in \cite{DScha1},   \cite{DScha2}, \cite{DScha3},  \cite{DScha4} and \cite{GRS}.\smallskip

The main contribution of the FTAP theorems is the  link between the concept of arbitrage and the pricing technique which is deduced. It is now very well known that the super-hedging prices of a European claim  are dually identified through the risk-neutral probability measures characterizing the NA condition. We may notice that the NA condition has been suitably chosen in the models of consideration in such a way that the set of all attainable claims is closed, see \cite[Theorem 2.1.1]{KS}. This allows one to apply the Hahn-Banach separation theorem, see \cite{DScha5}, and obtain dual elements that characterize the super-hedging prices. This is also the case for financial models with proportional transaction costs, see \cite[Section 3]{KS} and the references mentioned therein. \smallskip

The growing use of risk-measures in the context of the Basel banking supervision naturally calls into question the definition of the  super-hedging condition which is commonly accepted in the usual literature. Recall that a portfolio process $(V_t)_{t\in [0,T]}$ super-replicates a contingent claim $h_T$ at the horizon date $T>0$ means that $V_T\ge h_T$ a.s.. In practice, this inequality remains difficult to achieve and practitioners accept to take a moderate risk, choosing for example $\alpha\in (0,1)$ small enough so that $P(V_T-h_T\ge 0)\ge 1-\alpha$ is close to $1$. This is the case when considering the Value At Risk measure, see \cite{Jorion}, and we say that $V_T-h_T$ is acceptable. More generally, $V_T-h_T$ is said acceptable for a risk-measure $\rho$ if $\rho(V_T-h_T)\le 0$, see \cite{Acc}, \cite{Del1}, \cite{Del2}, \cite{DS},\cite{FP} and \cite{KL} for frictionless markets and \cite{AHR}, \cite{BL},  \cite{FR}, \cite{JMT}   for conic models. The acceptable positions play the role of the almost surely non negative random variables and allow one to take risk controlled by the risk measure we choose.  Moreover, by considering a larger family of acceptable positions, the hedging prices may be lowered as shown in \cite{P} for the Black and Scholes models with proportional transaction costs, see also the discussion in \cite{LM}. \smallskip

Pricing with  a coherent risk-measure has been explored and developed by Cherny in two major papers \cite{Ch1} and \cite{Ch2} for coherent risk-measures defined on the space of bounded random variables. Cherny supposes that the risk-measure $\rho$ (or equivalently the utility measure $u=-\rho$) is defined by a weakly compact determining set $\mathcal{D}$ of equivalent probability measures, i.e. such that $\rho(X)=\sup_{Q\in \mathcal{D}}E_Q(-X)$ for any $X\in L^{\infty}$. This representation  automatically holds for coherent risk-measures defined on  $L^{\infty}$. This motivates the choice of Cherny to suppose such a representation for the risk-measures  he considers  on $L^0$ as he claims that {\sl it is hopeless to axiomatize the notion of a risk measure on $L^0$ and then to obtain the corresponding representation theorem}, see \cite{Ch2}.\smallskip

Actually, the recent paper \cite{ZL} proposes an axiomatic construction of a dynamic coherent risk-measure on $L^0$ from the set of all acceptable positions. We consider such a dynamic risk-measure and we define the discrete-time portfolio processes as the processes  $(V_t)_{t\le T}$ adapted to a filtration $(\cF_t)_{t\le T}$ such that $V_t+\theta_{t}\Delta S_{t+1}-V_{t+1}$ is acceptable at time $t$ for some $\cF_t$-measurable strategy $\theta_{t}\in L^0(\R^d,\cF_t)$. This is a generalization of the classical definition where, usually, acceptable means non negative so that $V_t+\theta_{t}\Delta S_{t+1}\ge V_{t+1}$ almost surely. We then introduce a no-arbitrage condition we call NA as in the classical literature 	and we show that it coincides with the usual NA condition if the acceptable positions are the non negative random variables.  This NA condition allows one to dually characterize the super-hedging prices, at least when $\rho$ is time-consistent. One of our main contribution is a version of the Fundamental Theorem of Asset Pricing in presence of a risk-measure.\smallskip

Similarly, Cherny proposes in his papers \cite{Ch1} and \cite{Ch2} a no-arbitrage condition {\sl No Good Deal} (NGD) which is the key point to define the super-hedging prices. The approach is a priori slightly different: The NGD condition holds if there is no bounded  claim $X$ attainable from the zero initial capital such that $\rho(X)<0$. In our setting, the NA condition is formulated from the minimal  price super-hedging the zero claim, which is supposed to be non negative under NA. Clearly, there is a link between the NA and the NGD condition as $\rho(X)$ appears to be a possible super-hedging price for the zero claim. Actually, the NGD and the NA conditions are equivalent in the setting of Cherny, see Corollary \ref{NGD-NA}. Although,  in our paper we do not need to suppose the existence of a priori  given probability measure representing the risk-measure. This is why the proof of the FTAP theorem we formulate is more challenging as we cannot directly use an immediate compactness argument as done in \cite{Ch2} to obtain a risk-neutral probability measure. We then deduce a dual representation of the super-hedging prices in the case where the risk-measure is time-consistent. Under NA, we show that  the set of all risk-hedging prices is closed.  At last, we  formulate a dual representation  for a risk-measure defined on the whole set $L^0$, which is also a new contribution.

\section{Framework}\label{DynRM}

In discrete-time, we consider a    stochastic basis $(\Omega,\cF:=(\cF_t)_{t=0}^T,\bP)$  where the complete \footnote{This means that the $\sigma$-algebra contains the negligible sets so that an equality between two random variables is understood up to a negligible set. } $\sigma$-algebra $\cF_t$ represents the  information of the market available at time $t$. For any $t\le T$, $L^0(\R^d,\cF_t)$, $d\ge 1$, is the  space of all $\R^d$-valued random variables which are $\cF_t$-measurable, and endowed with the topology of convergence in probability. Similarly, $L^{p}(\R^d,\cF_t)$, $p\in [1,\infty)$ (resp. $p=\infty$), is the normed space of all $\R^d$-valued random variables which are $\cF_t$-measurable and admit a moment of order $p$ under the probability measure $\bP$ (resp. bounded).  In particular, $L^{p}(\R_{+},\cF_t)=\{X\in L^{p}(\R,\cF_t)|X\ge 0\}$ and $L^{p}(\R_{-},\cF_t)=-L^{p}(\R_{+},\cF_t)$ when $p=0$ or $p\in [1,\infty]$. All equalities and inequalities between random variables are understood to hold everywhere on $\Omega$ up to a negligible set. If $A_t$ is a set-valued mapping (i.e. a random set of $\R^d$), we denote by $L^0(A_t,\cF_t)$ the set of all $\cF_t$-measurable random variables $X_t$ such that 
$X_t\in A_t$ a.s..  We say that $X_t\in L^0(A_t,\cF_t)$ is a measurable selection of $A_t$. In our paper, a random set $A_t$ is said $\cF_t$-measurable if it is graph-measurable, see \cite{Mol}, i.e.
$${\rm Gr\,} A_t=\{(\omega,x)\in\Omega\times \R^d:~x\in A_t(\omega)\}\in \cF_t\times \cB(\R^d).$$

It is well known that $L^0(A_t,\cF_t)\ne \emptyset$ if and only if $A_t\ne \emptyset$ a.s., see \cite[Th.~4.4]{hes02}. When referring to this property, we shall say that we use a "measurable selection argument" as it is usual to say.  \smallskip

We consider a  dynamic coherent risk-measure $X\mapsto (\rho_t(X))_{t\le T}$  defined on the space $L^0(\overline{\R},\cF_T)$, $\overline{\R}=[-\infty,\infty]$. Precisely, we consider the risk-measure of \cite{ZL}, where  an extension to the whole space $L^0(\overline{\R},\cF_T)$ is proposed. Recall that, in this paper, the risk-measure is  constructed from its $L^0$-closed acceptance sets $(\mathcal A_t)_{t\le T}$ of  acceptable financial positions $\mathcal A_t$ at time $t\le T$. We suppose that $\mathcal A_t$ is a closed convex cone. In the following, we use the  conventions:

\bean 0\times (\pm \infty)=0, \quad (0,\infty)\times (\pm \infty)=\{\pm\infty\},\\
\R+(\pm \infty)=\pm\infty,\quad \infty-\infty=-\infty+\infty=+\infty.
\eean

 For $X\in L^0( \overline{\R},\cF_T)$, $\rho_t(X)$ may be infinite and $\rho_t(X)\in \R$ a.s. if and only if $X\in {\rm Dom\,} \mathcal A_t$ where 
 \bean {\rm Dom\,} \mathcal A_t&:=&\{X\in L^0(\R,\cF_T):~\mathcal A_t^X\ne \emptyset \},\\
\mathcal A_t^X&:=&\{C_t\in L^{0}(\R,\cF_t)|\,X+C_t \in \mathcal A_t\}.\eean

Actually, we have $\rho_t(X)=\essinf_{\cF_t}\mathcal A_t^X$ if $X\in {\rm Dom\,} \mathcal A_t$. Recall that the following properties hold (see \cite{ZL}):
\begin{prop}\label{pro}
The risk-measure $\rho_t$ satisfies the following properties:\smallskip

Normalization: $\rho_t(0)=0$; \smallskip

Monotonicity: $\rho_t(X)\ge \rho_t(X')$ whatever  $X, X' \in L^0( \overline{\R},\cF_T)$  s.t. $X\le X'$;\smallskip

Cash invariance: $\rho_t(X+m_t)=\rho_t(X)-m_t$ if $m_t\in L^{0}(\R,\cF_t)$, and

  $X \in L^0( \overline{\R},\cF_T)$;\smallskip

Subadditivity: $\rho_t(X+X')\le \rho_t(X)+\rho_t(X')$ if  $X, X' \in L^0( \overline{\R},\cF_T)$ ;\smallskip

Positive homogeneity: $\rho_t(k_tX)=k_t\rho_t(X)$ if $k_t\in L^{0}(\R_+,\cF_t)$, $X \in L^0( \overline{\R},\cF_T)$. \smallskip

Moreover,  $\rho_t$ is lower semi-continuous i.e., if $X_n\to X$ a.s., then $\rho_t(X)\le \liminf_n \rho_t(X_n)$ a.s., and we have \begin{equation}\label{Acset}
\cA_t=\{X\in {\rm Dom\,} \mathcal A_t\,|\,\rho_t(X)\le 0\}.
\end{equation}
\end{prop}

In the following, we define $\cA_{t,u}:=\cA_t\cap L^0( \overline{\R},\cF_u)$ for $u\in [t,T]$. Let $(S_t)_{ t\le T}$ be a  process describing the discounted prices of $d$ risky assets such that $S_t\in L^{0}(\R^d_{+},\cF_t)$ for any $t\ge 0$. A contingent claim with maturity date  $t+1$ is defined by a real-valued $\cF_{t+1}$-measurable random variable $h_{t+1}$. In the paper \cite{ZL}, the super-hedging problem for the payoff $h_{t+1}$ is solved with respect to the dynamic risk-measure $(\rho_t)_{ t\le T}$. Precisely:
\begin{defi}\label{pT-1}A payoff $h_{t+1}\in L^0(\R,\cF_{t+1})$ is said to be risk-hedged at time $t$ if there exists $P_{t}\in L^{0}(\R,\cF_{t})$ and a strategy $\theta_{t}$ in $L^{0}(\R^d,\cF_{t})$ such that $P_{t}+\theta_{t}\Delta S_{t+1}-h_{t+1}$ is acceptable at time $t$. In that case, we say that $P_{t}$ is a risk-hedging price.
\end{defi}

Let $\cP_{t}(h_{t+1})$ be the set of  all risk-hedging prices $P_{t}\in L^{0}(\R,\cF_{t})$ at time $t$ as in Definition \ref{pT-1}. In the following, we suppose that $\cP_{t}(h_{t+1})\ne \emptyset$.   This is the case if there    exist $a_{t},b_{t}\in L^0(\R,\cF_{t})$ such that $h_{t+1}\le a_{t}S_{t+1}+b_{t}$. This inequality trivially  holds for European call and put options.\smallskip

\begin{defi}
The minimal risk-hedging price of the contingent claim $h_{t+1}$ at time $t$ is defined as
\begin{equation}
P_{t}^*:=\mathop{\essinf}\limits_{\theta_{t}\in L^{0}(\R^d,\cF_{t})}\cP_{t}(h_{t+1}).
\end{equation}
\end{defi}
Note that the minimal risk-hedging price $P_{t}^*$ of $h_{t+1}$ is not necessarily  a price, i.e.  it is not necessarily  an element of $\cP_{t}(h_{t+1})$ if this set is not closed. One contribution of our paper is to study a no-arbitrage condition under which  $P_{t}^*\in \cP_{t}(h_{t+1})$. \smallskip

Starting from the contingent claim $h_T$ at time $T$, we  recursively define
$$P_{T}^*:=h_T \;,\; P_{t}^*:=\mathop{\essinf}\limits_{\theta_{t}\in L^0(\R^d,\cF_{t})}\cP_{t}(P_{t+1}^*),$$
where $P_{t+1}^*$ may be interpreted as a contingent claim $h_{t+1}$. The interesting question is whether $P_{t}^*$ is actually a price, i.e. an element of $\cP_{t}(P_{t+1}^*)$, or equivalently whether $\cP_{t}(P_{t+1}^*)$ is closed. In the classical setting, recall that closedness is obtained under the NA condition. \smallskip

\begin{defi} A stochastic process $(V_t)_{t\le T}$ adapted to $(\cF_t)_{t\le T}$, starting from an initial endowment $V_{0}$ is a portfolio process if, for all $t\le T-1$, there exists $\theta_{t}\in L^0(\R^d,\cF_t)$ such that $V_t+\theta_{t}\Delta S_{t+1}-V_{t+1}$ is acceptable at time $t$. Moreover, we say that it super-hedges the payoff $h_T\in L^0([-\infty,\infty],\cF_T)$ if $V_T\ge h_T$ a.s..

\end{defi}

Note  that $V_{T-1}+\theta_{T-1}\Delta S_{T}-V_{T}$ is supposed to be acceptable at time $T-1$. Therefore, $V_T\ge h_T$ implies that $V_{T-1}+\theta_{T-1}\Delta S_{T}-h_{T}$ is  acceptable at time $T-1$. In the following, we actually set $V_T=h_T$ where  $h_T\in L^0(\R,\cF_T)$ is a European claim. Notice that, if $P_{T-1}^*=-\infty$ on some non null set, then, the one step pricing procedure of \cite{ZL} may be applied as  the risk-measure is defined on $L^0([-\infty,\infty],\cF_T)$. Actually, this is trivial to super hedges $P_{T-1}^*=-\infty$ by $P_{T-2}^*=-\infty$. This means that the  backward procedure of \cite{ZL} may be  applied without any no-arbitrage condition. Let us now recall this procedure. \smallskip

We  define  $P_{T}^*=h_T=:h$ and let us consider the set $\cP_{t}(P_{t+1}^*)$ of all prices $p_t$ at time $t$ allowing one to start a portfolio strategy $\theta_{t}\in L^0(\R^d,\cF_{t})$ such that $p_t+\theta_t\Delta S_{t+1}=P_{t+1}^*+a_{t,t+1}$ where $a_{t,t+1}\in L^0(\R,\cF_t)$ is an acceptable position at time $t$. This is a generalization of the classical super-hedging inequality $p_t+\theta_t\Delta S_{t+1}\ge P_{t+1}^*$. We have 
$$\cP_{t}(P_{t+1}^*)=\{\theta_{t} S_{t}+\rho_{t}(\theta_{t} S_{t+1}-P_{t+1}^*):\theta_{t}\in L^0(\R^d,\cF_{t})\}
+L^{0}(\R^d_+,\cF_{t}),$$
and, recursively, we define:
$$P_{t}^*=\mathop{\essinf}\limits_{\theta_{t}\in L^0(\R^d,\cF_{t})}\cP_{t}(P_{t+1}^*).$$
In \cite{ZL}, a jointly measurable version of the random function $g_t$ that appears above in the characterization of $\cP_{t}(P_{t+1}^*)$, i.e.
\bea \label{gt} g_t^h(\omega,x):=x S_{t}+\rho_{t}(x S_{t+1}-P_{t+1}^*),\eea
is constructed in the one-dimensional case. With the same arguments, we may obtain a jointly measurable version of $g_t^h(\omega,x):= x S_{t}+\rho_{t}(x S_{t+1}-P_{t+1}^*) $ if $x\in\R^d$. Moreover, by similar arguments, we also show that $P_{t}^*=\inf\limits_{x\in \R^d} g_t^h(x).$ 

Let $V$ be a portfolio process with  $V_T=h_T=h$. By definition, we have that  $\rho_{T-1}(V_{T-1}+\theta_{T-1}\Delta S_{T}-h_T)\le 0$. We deduce that   $V_{T-1}\ge P_{T-1}^*$ and, by induction, we get that  $V_t\ge P_{t}^*$ for all $t\le T$, since $V_t$ is a  risk-hedging price for $V_{t+1 }\ge P_{t+1}^*$ at time $t+1$. In particular, $V_t\in \cP_{t}(P_{t+1}^*)\ne \emptyset$ for all $t\in T-1$.

\section{ No-arbitrage and pricing with risk-measures}

An instantaneous profit is the possibility to super-replicate the zero contingent claim  at a negative price, see \cite{BCL}, \cite{CL}.

\begin{defi}
Absence of Instantaneous Profit (AIP) holds if, for any $ t\le T$,
\begin{equation}\label{NGD}
\cP_t(0)\cap L^{0}(\R_{-},\cF_t)=\{0\}.
\end{equation}

\end{defi}

It is clear that AIP  holds  at time $T$ since $\cP_T(0)=L^{0}(\R_{+},\cF_T)$. We now formulate characterizations of the AIP condition  in the multi-dimensional setting. We denote by $S(0,1)$ the set of all $z\in \R^d$ such that $|z|=1$. We present our first result:

\begin{theo} \label{AIP-theo}
 The following statements are equivalent:
\begin{enumerate}
    \item AIP holds between time $t-1$ and $t$.\label{1} \smallskip
    
    \item $\rho_{t-1}(x\Delta S_t)\geq 0$,\, for any $x\in\R^d$,\,{\rm a.s.}. \label{2}\smallskip
    
    \item $\rho_{t-1}(z\Delta S_t)\geq 0$,\, for any $z\in S(0,1)$,\,{\rm a.s.}.  \label{3}\smallskip
    
    \item Let $x_{t-1}\in L^0(\R^d,\cF_{t-1})$. If $x_{t-1}\Delta S_t$ is acceptable on some non null set $F_{t-1}\in \cF_{t-1}$, then $\rho_{t-1}(x_{t-1}\Delta S_t) = 0$ on $F_{t-1}$. \label{4}
\end{enumerate}
\end{theo}
\begin{proof}
$\ref{1}\Longleftrightarrow \ref{2}$. Consider $h_t = 0$ under AIP. As $P_{t-1}^*=\inf\limits_{x\in \R^d} g_{t-1}^0(x)\ge 0$, we deduce that, for all $x\in\R^d$, $g_{t-1}^0(x) = xS_{t-1} + \rho_{t-1}(xS_t) = \rho_{t-1}(x\Delta S_t) \geq 0$.

The equivalence $\ref{2}\Longleftrightarrow\ref{3}$ is clear by homogeneity. Let us show that $\ref{2}\implies \ref{4}$. Suppose that $x_{t-1}\Delta S_t$ is acceptable on $F_{t-1}$, i.e. $\rho_{t-1}(x_{t-1}\Delta S_t) \leq 0$ on $F_{t-1}$. Then, by \ref{2}, we  have $\rho_{t-1}(x_{t-1}\Delta S_t) = 0 \ \text{on} \ F_{t-1}$. Let us show that $\ref{4}$ implies $\ref{2}$.  Consider the set $F_{t-1}=\{\rho_{t-1}(x_{t-1}\Delta S_t)<0\}\in \cF_{t-1}$. Then, $x_{t-1}\Delta S_t$ is acceptable on $F_{t-1}$ hence by \ref{4}, $\rho_{t-1}(x_{t-1}\Delta S_t)=0$ on $F_{t-1}$, which implies that $P(F_{t-1})=0$. Therefore, $\rho_{t-1}(x_{t-1}\Delta S_t)\geq 0$ a.s.. \end{proof}

In the following, we consider a contingent claim $h_t\in L^0(\R,\cF_t)$ and a  jointly measurable version (see \cite{ZL}) of the random function
\bea \label{gt-1} g_{t-1}(\omega,x):=x S_{t-1}(\omega)+\rho_{t-1}(x S_{t}-h_t)(\omega)\eea
which is associated to $h_t$.  We then introduce two types of no-arbitrage conditions we comment below.

\begin{defi}
We say that the Symmetric Risk Neutral condition SRN holds at time $t$ if, almost surely, for any $z_t\in L^0(S(0,1),\cF_t)$, $\rho_{t}(z_t\Delta S_{t+1})=0$ if and only if $\rho_{t}(-z_t\Delta S_{t+1})=0$. We say that SRN holds if it holds at any time. 
\end{defi}

Observe that the SRN condition means that a zero cost position $z_t$ is risk-neutral if and only if $-z_t$ is risk neutral.

\begin{defi}\label{defi-SAIP} We say that the no-arbitrage NA condition holds at time $t$ when both conditions  AIP and SRN hold at time $t$. We say that NA holds if it holds at any time.
\end{defi}

Note that the NA condition depends on the risk-measure. In the usual case where $\rho_{t}(X)=-\essinf_{\cF_{t}}X$ or, equivalently, there is no risk measure in the sense that the acceptable positions are the non-negative random variables, then the NA condition above coincides with the usual one as claimed in the following new result, see the proof in Appendix:

\begin{prop} \label{na-classical na} Suppose that the risk-measure is $\rho_{t}(X)=-\essinf_{\cF_{t}}X$. Then, the NA condition coincides with the classical NA condition of frictionless models, i.e. it is equivalent to the existence of a risk-neutral probability measure.
\end{prop}

 We recall  that a function $f: \Omega\times\R^d\to \overline{\R}$ is an $\mathcal{F}_t$-normal integrand,  if its epigraph is $\mathcal{F}_t$-measurable and closed. Since the probability space is complete, we know by  \cite[Corollary 14.34]{RW} that it is equivalent to suppose that $f(\omega,x)$ is $\mathcal{F}_t\otimes\mathcal{B}(\R^d)$-measurable and lower semi-continuous (l.s.c.) in $x$. Moreover, by \cite[Theorem 14.37]{RW}, we have:
\begin{prop}\label{measurability}
If $f$ is an $\mathcal{F}_t$-normal integrand,   $\inf_{y\in\R^d}f(\omega,y)$ is $\cF_t$-measurable and 
$\{(\omega,x)\in\Omega\times\R^d: f(\omega,x) = \inf_{y\in\R^d}f(\omega,y)\}\in\mathcal{F}_t\otimes\mathcal{B}(\R^d)$ is a measurable closed set.
\end{prop}
As we may choose a jointly measurable version of $g_t(\omega,x)$ when the payoff is $h_{t+1}=0$,  we consider a jointly measurable version of $\rho_t(\omega,x):= \rho_t(x\Delta S_{t+1})$ i.e. $\rho_t(\omega,x)$ is $\mathcal{F}_t\otimes\mathcal{B}(\R^d)$-measurable. Then, $\rho_t$ is an $\mathcal{F}_t$-normal integrand. By Proposition \ref{measurability}, the set $\Gamma_t = \{z:~ \rho_t(z\Delta S_{t+1}) = \inf_{y\in S(0,1)}\rho_t(y\Delta S_{t+1})\}$ is $\mathcal{F}_t$-measurable. Moreover, each $\omega$-section of $\Gamma_t$ is non empty since $\rho_t$ is l.s.c. and $S(0,1)$ is compact. Therefore, by a measurable selection argument, we may select $z_t\in L^0(S(0,1),\mathcal{F}_t)$ such that $\rho(z_t\Delta S_{t+1}) = \inf_{z\in S(0,1)}\rho_t(z\Delta S_{t+1})$ a.s.. \smallskip

Our first contribution is to show that, under NA,  infimum super-hedging prices are minimal prices. To do so, we need the following new results which  are proved  in Appendix.

\begin{theo} \label{constant_g} Suppose that AIP holds and consider $z_{t-1}\in L^0(S(0,1),\mathcal{F}_{t-1})$. Then, on the set $H_{t-1}= \{\rho_{t-1}(z_{t-1}\Delta S_t) = 0\}\cap\{\rho_{t-1}(-z_{t-1}\Delta S_t) = 0\}$, the random  mapping
$x\mapsto g_{t-1}(\omega,x)$ given by (\ref{gt-1}) is a.s.  constant on the line $\R z_{t-1}$, i.e. $g_{t-1}(\omega,x_1)=g_{t-1}(\omega,x_2)$ for all $x_1,x_2\in \R z_{t-1}(\omega)$ and $\omega\in H_{t-1}$.
\end{theo}

\begin{prop}\label{minimum_g} Let $h_t\in L^0(\R,\cF_t)$ be a payoff  such that $\rho_{t-1}\left(h_t\right)<\infty$ a.s..Consider the random function $g_{t-1}$ associated to $h_t$ given by (\ref{gt-1}). For any $z_{t-1}\in L^0(S(0,1),\mathcal{F}_{t-1})$, consider the random set 
$$F_{t-1} = \{\rho_{t-1}(z_{t-1}\Delta S_t) >0\}\cap \{\rho_{t-1}(-z_{t-1}\Delta S_t) >0\}.$$ We   have:
\begin{align*}
    \lim_{|r|\to\infty} g_{t-1}(\omega,rz_{t-1}) = +\infty, \quad \forall \omega\in F_{t-1}.
\end{align*}
hence $g_{t-1}$ admits a minimum on the line $\R z_{t-1}$ when $\omega\in F_{t-1}$.
\end{prop}

\begin{theo}\label{SAIP-price} Let $h_t\in L^0(\R,\cF_t)$ be  s.t. $\rho_{t-1}\left(h_t\right)<\infty$. Consider the  function $g_{t-1}$ associated to $h_t$ given by (\ref{gt-1}). Suppose that $z_{t-1}\in L^0(S(0,1),\mathcal{F}_{t-1})$ is such that  $$\rho_{t-1}(z_{t-1}\Delta S_{t}) = \inf_{z\in S(0,1)}\rho_{t-1}(z\Delta S_{t}).$$  Then, on the set $F_{t-1} = \{\rho_{t-1}(z_{t-1}\Delta S_{t})>0\}\cap \{\rho_{t-1}(-z_{t-1}\Delta S_{t})>0\},$ the random function $g_{t-1}$ admits a minimum.
\end{theo}

The following theorem is our first main contribution and shows that the set of all risk-hedging prices is closed under NA. This is a consequence of Theorems 3.7 and 3.8 that insures that $g_{t-1}$ admits a global minimum, see Theorem \ref{SAIP-price}.

\begin{theo}\label{minimal hedging price}
Suppose that $NA$ holds at  time $t\leq T$ and consider a payoff $h_{t+1}\in L^0(\R,\cF_{t+1})$ such that $\rho_t(h_{t+1})<\infty$ a.s.. Then, the minimal risk-hedging price $P^*_t$ for the payoff $h_{t+1}$  is a price.
\end{theo}
Notice that the proof of the theorem  above (see Appendix) provides the existence of an optimal hedging strategy $\theta_{t}^*\in L^0(\R,\cF_{t})$ such that $$P_{t}^*=g_{t}(\theta_{t}^*)=\theta_{t}^* S_{t}+\rho_{t}(\theta_{t}^* S_{t+1}-h_{t+1})\in \cP_{t}(h_{t+1}).$$

In the following, we say that a payoff $h_{t+1}$ is  not freely attainable at time $t$ if  it satisfies $\rho_t(-h_{t+1}) >0$ a.s. and $\rho_t(h_{t+1})<\infty$ a.s..  Note that if $\rho_t(-h_{t+1}) >0$, then it is not possible to get the payoff $h_{t+1}$  from nothing when writing $0=h_{t+1}+(-h_{t+1})$ and letting aside  $(-h_{t+1})$ since the latter is not acceptable. Notice that, if $\rho_t(X)=-\essinf_{\cF_t}(X)$ as in the usual case, $\rho_t(-h_{t+1}) >0$ means that $\esssup_{\cF_t}(h_{t+1})>0$ and recall that $h_{t+1}$  is acceptable if $h_{t+1}\ge 0$ a.s..   The following theorem gives an interpretation of the NA condition. Precisely, NA means that the price of any no freely attainable and acceptable payoff is strictly positive. In the usual case, a no freely attainable and acceptable payoff is a non negative payoff which does not vanish on a non null $\cF_t$-measurable set. \smallskip

We then have a new financial interpretation of the NA condition, as proved in Appendix:

\begin{theo}\label{SAIP} The NA condition holds at time $t\le T$ if and only if the infinimum risk-hedging price $P_t^*$ of any  no freely attainable  and acceptable payoff $h_{t+1}$ at time $t$  is strictly positive. Moreover,  under NA, the infimum risk-hedging price $P_{t}^*$ of any contingent claim $h_{t+1}\in L^0(\R^d,\cF_{t+1})$ satisfies

 $$\rho_{t}(-h_{t+1})\ge P^*_t\ge -\rho_{t}(h_{t+1}).$$

\end{theo}

\section{FTAP and dual representation for time-consistent risk measures. }\label{consistency}

\begin{defi} A dynamic risk-measure $(\rho_t)_{ t\le T}$ is said time-consistent if $\rho_{t+1}(X)=\rho_{t+1}(Y)$ implies $ \rho_{t}(X)=\rho_{t}(Y)$ for  $X, Y \in L^0(\R,\cF_T)$ and $ t \le T-1$ (see Section 5 in \cite{DS}).
\end{defi}

The following result is very well known, see \cite{Acc}.

\begin{lemm}\label{TC-acc}
A dynamic risk-measure $(\rho_t)_{ t\le T}$ is time-consistent if and only if its family of acceptable sets $(\cA_t)_{t\le T}$ satisfies
\bea
\cA_{t,T}=\cA_{t,t+1}+\cA_{t+1,T},\,\forall t\le T-1.
\eea
\end{lemm}

Observe that, if $(\rho_t)_{ t\le T}$ is  time-consistent, we may show by induction that $\rho_{t}(-\rho_{t+s}(\cdot))=\rho_{t}(\cdot)$ for any $t\le T$ and $s\ge 0$ such that $s+t\le T$. In the following, we introduce another possible definition for the risk-hedging prices in the multi-period model, where the risk is only measured at time $t$.
\begin{defi}
The contingent claim $h_T\in L^0(\R,\cF_T)$ is said   directly risk-hedged at time $t\le T-1$ if there exists a (direct) price $P_{t}\in L^{0}(\R,\cF_{t})$ and a strategy  $(\theta_{u})_{u=t}^{T-1}$ such that that $P_{t}+\sum\limits_{t\le u\le T-1}\theta_{u}\Delta S_{u+1}-h_T$ is acceptable at time $t$.
\end{defi}

The set of all direct risk-hedging prices at time $t$ is then given by
\bean \bar\cP_t(h_T)&=&\left\{\rho_t\left(\sum\limits_{t\le u\le T-1}\theta_{u}\Delta S_{u+1}-h_T\right):~\theta_{u}\in L^0(\R^d,\cF_u)\right\}+ L^0(\R,\cF_{t}).\eean
and  the infimum direct risk-hedging price is
$$\bar P_{t}^*(h_T):=\mathop{\essinf}\limits_{(\theta_{u})_{u=t}^{T-1}}\bar\cP_t(h_T).$$

The following result is proved in \cite{ZL} and shows that the direct infimum risk-hedging prices may coincide with  the infimum prices derived from the step by step backward procedure  developed before, i.e. such that 
$$P_{t}^*(h_T)=\mathop{\essinf}\limits_{\theta_{t}\in L^0(\R,\cF_{t})}\cP_{t}(P_{t+1}^*(h_T)),$$ where $P_{T}^*(h_T)=h_T$.
\begin{theo}
Suppose that  the dynamic risk-measure $(\rho_t)_{ t\le T}$ is time-consistent. Then,  $\bar P_{t}^*(h_T)=P_{t}^*(h_T)$ for any $ t\le T-1$. Moreover, the direct infimum risk-hedging prices are direct prices if and only if the infimum prices of the backward procedure are prices. 
\end{theo}

\begin{coro}\label{CoroConsistency} Suppose that  the dynamic risk-measure $(\rho_t)_{ t\le T}$ is time-consistent. Then,  $\bar\cP_{t}(h_T)=\cP_{t}(h_T)$ for all $t\le T$.

\end{coro}

\subsection{Dual representation}

As mentioned by Cherny \cite[Theorem 2.2]{Ch2} and shown in \cite{Del2}, any time-consistent risk-measure $\rho_t$ at time $t$, restricted to the set of all bounded random variables, is characterized by a family $\mathcal{D}_t$ of absolutely continuous probability measures such that $\rho_t(X)=\esssup_{ Q\in \mathcal{D}_t}E_Q(-X|\cF_t)$.  In the following, we consider the risk-measure $\rho$ on $L^0$ as defined in this paper. The goal is to understand whether it is possible to get a dual characterization of $\rho$ on the whole set $L^0$, at least under some conditions.  For $X\in L^0$, we define $E_Q(-X|\cF_t)$ as $E_Q(-X|\cF_t)=E_Q(X^-|\cF_t)-E_Q(X^+|\cF_t)$ with the convention $\infty-\infty=\infty$. We say that a random variable $X$ is $\cF_t$-bounded from above if $X\le c_t$ a.s. for some $c_t\in L^0(\R_+,\cF_t)$. The proofs of the following new contributions are postponed in Appendix. They provide a dual representation of the risk-measure.

\begin{prop}\label{propRP0} Let $(\rho_t)_{t=0,\cdots,T}$ be the coherent risk-measure as defined in Section \ref{DynRM}. Then, there exists 
a family $\mathcal{D}_t$ of absolutely continuous probability measures such that, for every $\cF_t$-bounded from above  random variable $X$, we have:
\bea \rho_t(X)=\esssup_{ Q\in \mathcal{D}_t}E_Q(-X|\cF_t). \label{RP}\eea
\end{prop}

Unfortunately, it is unrealistic to expect that (\ref{RP}) may be extended in general from $L^{\infty}$ to $L^0$, as mentioned by Cherny, \cite{Ch2}. The main problem is about the non negatives random variables as we shall see in the proof of the next proposition. Before, let us see a trivial  example where we may meet some difficulties for non negative random variables. 

\begin{ex} We consider $\Omega=[0,1]$ equipped with the Borel $\sigma$-algebra and the Lebesgue measure $P$. The random variable $X(\omega)=\omega^{-1}1_{(0,1]}(\omega)$ is non negative hence acceptable. Let us define the  acceptable positions as the closure in $L^0$ of the  random variables  $Z$ such that $E_P(Z)=E_P(Z^+)-E_P(Z^-)\ge 0$. We then define $\rho$ on $L^0$ as in Section 2, see \cite{ZL}. As $E_P(X)=\infty$, we deduce that $Z_{\alpha}:=X-\alpha$ is acceptable for all $\alpha>0$. On the other hand, $P(Z_{\alpha}<0)=1-\alpha^{-1}$ tends to $1$ as $\alpha\to \infty$, which is unrealistic if $Z_{\alpha}$ is acceptable. 

Consider $Q\in  \mathcal{D}_0$ and $Y=dQ/dP$. Suppose that $P(Y>1)>0$. We then choose $\alpha<0$ and $\beta>0$ such that $\alpha P(Y>1)+\beta P(Y\le 1)=0$. Then, $X=\alpha 1_{\{Y>1\}}+\beta 1_{\{Y\le 1\}}$ is acceptable as $E_P(X)=0$. Therefore, by  (\ref{RP}), $E_Q(X)\ge 0$. Actually, 
\bean E_Q(X)=E_P(XY)=E_P(\alpha Y1_{\{Y>1\}}+\beta Y1_{\{Y\le 1\}})\le E_P(X)=0\eean and
$E_Q(X)=0$ if and only if $\alpha Y1_{\{Y>1\}}+\beta Y1_{\{Y\le 1\}}=X$. In that case $Y=1$ on $\{Y>1\}$ hence a contradiction. We deduce that $Y\le 1$ a.s.. At last, since $Y\le 1$ a.s., we deduce that $Y=1$ a.s.. We then deduce that $\mathcal{D}_0=\{P\}$.

\end{ex}

In the following, we denote by $\mathcal A_t^{\infty,+}$ the set of all acceptable positions at time $t$ which are $\cF_t$-bounded from above.

\begin{prop}\label{DR}
Suppose that $\mathcal A_t$ is the closure of $\mathcal A_t^{\infty,+}+L^0(\R^+,\cF_T)$ in $L^0$ and assume that, for some fixed $\varepsilon>0$, $\mathcal A_t^{\infty,+}$ contains all the random variables $Z$ which are $\cF_t$-bounded from  above and satisfy $P(Z<0)\le \varepsilon$.  Let $(\rho_t)_{t=0,\cdots,T}$ be the coherent risk-measure as defined in Section \ref{DynRM}. Then, there exists 
a family $\mathcal{D}_t$ of absolutely continuous probability measures such that we have.
\bea \rho_t(X)=\esssup_{ Q\in \mathcal{D}_t}E_Q(-X|\cF_t),\quad \forall X\in L^0. \label{RP1}\eea
\end{prop}

The proof of the proposition above (see Appendix) shows that (\ref{RP1}) holds as soon as it holds for any acceptable position which is the sum of an $\cF_t$-bounded position plus a non negative one.  By Proposition \ref{propRP0}, (\ref{RP1}) holds  for any $\cF_t$-bounded position. Therefore, the difficulty in proving (\ref{RP1})  stems from the non negative random variables. 

\subsection{FTAP and dual description of the risk-hedging prices}

We consider the set of all attainable claims $\cR_t^T$ between $t$ and $T$, when starting from the zero initial endowment, i.e. 
$$\cR_{t,T}:=\left\{ \sum_{u=t+1}^T\theta_{u-1}\Delta S_u:~\theta_u\in L^0(\R^d,\cF_u),\,u\ge t  \right\}.$$

We observe that $\bar\cP_{t}(0)=\left(\cA_{t,T}-\cR_{t,T}\right)\cap L^0(\R,\cF_t)$. In the following, we consider the sets $\cZ_{t,T}:=\cR_{t,T}-\cA_{t,T}$ and the sets  
$$\cA_{t,T}^0=\{ X\in L^0(\R,\cF_T):~\rho_t(X)=\rho_t(-X)=0\}.$$
\begin{rem}\label{remark acceptance set}
Note that $\cA_{t,T}^0 = \cA_{t,T}\cap(-\cA_{t,T})$. Indeed, first observe that $\cA_{t,T}^0 \subseteq \cA_{t,T}\cap(-\cA_{t,T})$. Reciprocally, if $x_{t,T}\in\cA_{t,T}\cap(-\cA_{t,T})$, we have:
\begin{align*}
    0 = \rho_t(x_{t,T} - x_{t,T}) \leq \rho_t(x_{t,T}) + \rho_t(-x_{t,T}) \leq 0.
\end{align*}
This implies $\rho_t(x_{t,T}) = \rho_t(-x_{t,T}) = 0$ hence $x_{t,T}\in\cA_{t,T}^0$.\smallskip

The set $\cZ_{t,T}$ is the family of all  claims that are attainable up to an acceptable position at time $t$ since every attainable claim $r_{t,T}\in  \cR_t^T$ may be written as $r_{t,T}=(r_{t,T}-a_{t,T})+a_{t,T}$ where $a_{t,T}\in \cA_{t,T}$  and $r_{t,T}-a_{t,T}\in \cZ_{t,T}$.
\end{rem}

We now formulate intermediate new results that we need to prove the FTAP theorem, which is the first contribution of this section.

\begin{theo}\label{ClosureNA}

Assume that the risk measure is time-consistent. Suppose that $\cR_{t,T}\cap \cA_{t,T}=\cA_{t,T}^0$ . Then, AIP holds and $\cZ_{t,T}$ is closed in $L^0$ for every $t\le T-1$. 
\end{theo}

\begin{theo} \label{NA-Formulation} Suppose that the risk-measure is time-consistent. Suppose that NA holds and $\cA_{t,T}\cap L^0(\R_-,\cF_T)=\{0\}$, for every $t\le T$ . Then, we have $\cZ_{t,T}\cap L^0(\R_+,\cF_T)=\{0\}$ and $\cR_{t,T}\cap \cA_{t,T}=\cA_{t,T}^0$ for every $t$.
\end{theo}

\begin{theo}[FTAP] \label{FTAP}
Suppose that the risk-measure is time-consistent and  $\cA_{t,T}\cap L^0(\R_-,\cF_T)=\{0\}$ for every $t\le T$. Then, the following statements are equivalent:
\begin{itemize}
\item [1)] NA\smallskip

\item[2)] $\cR_{t,T}\cap\cA_{t,T} = \cA_{t,T}^0$, for every $t\le T$.\smallskip

\item[3)] $\cR_{0,T}\cap\cA_{0,T} = \cA_{0,T}^0$.\smallskip

\item [4)] $\overline{\cZ_{t,T}}\cap \cA_{t,T} = \cA_{t,T}^0$, for every $t\le T$. \smallskip

\item [5)] $\overline{\cZ_{0,T}}\cap \cA_{0,T} = \cA_{0,T}^0$.\smallskip

\item [6)] $\cZ_{0,T}\cap \cA_{0,T} = \cA_{0,T}^0$ and $\cZ_{0,T}$ is closed in $L^0$.\smallskip

\item [7)]For all $t\le T-1$, there exists $Q=Q^t\sim P$ with $dQ/dP\in L^{\infty}((0,\infty),\cF_T)$ such that $(S_u)_{u=t}^T$ is a $Q$-martingale and, for all $t\le T-1$, for all $X$ such that $E_Q(X^-|\cF_t)<\infty$ a.s.,   $\rho_t(X)\ge -E_Q(X|\cF_t)$. \smallskip

Moreover, for all $x\in \cA_{t,T}\setminus\cA_{t,T}^0$, there exists such a $Q=Q_{x}^t$ such that $\bP(E_Q(x|\cF_t)\ne 0)>0$.
\end{itemize}
\end{theo}

The following result is the second main contribution of this section. It provides a dual description of the payoffs that can be super-hedged under NA. To do so, we denote by $\cQ^e_t$ (and $\cQ^e=\cQ^e_0$)  the set of equivalent  martingale measures $Q$ that satisfies $\rho_t(X) \geq -E_Q(X|\cF_t)$, for all $X$ such that $E_Q(X^-|\cF_t) < \infty$ a.s.. We have $\cQ^e_t\neq\emptyset$ under NA. We restrict the payoffs to the class $L_S(\R,\cF_T)$ of random variables $h_T\in L^0(\R,\cF_T)$ satisfying:
\begin{align*}
    |h_T| \leq c^0 + \sum_{i=1}^dc^iS_T^i,\quad P-a.s.
\end{align*}
for some constants $c^0,...,c^d$ that may depend on $h_T$.
\begin{theo}\label{SHP}
Suppose that the risk-measure is time-consistent and  we have $\cA_{t,T}\cap L^0(\R_-,\cF_T)=\{0\}$ for every $t\le T$. Consider the following sets:
\begin{align*}
    \Gamma_{0,T} &:=\cZ_{0,T} \cap  L_S(\R,\cF_T),\\
    \Theta_{0,T} &:= \left\{h_T\in L_S(\R,\cF_T), \sup_{Q\in\cQ^e}E_Q(h_T) \leq 0\right\}.
\end{align*}
Then, under the NA condition, $\Gamma_{0,T}=\Theta_{0,T}$ and the minimal risk-hedging price $P_0^*(h_T)$ of any contingent claim $h_T\in L_S(\R,\cF_T)$ is given by
\begin{align*}
    P_0^*(h_T) = \sup_{Q\in\cQ^e}E_Q(h_T).
\end{align*}
\end{theo}

\subsection{Comparison with the No Good Deal condition}

We recall that the No Good Deal condition (NGD) of Cherny  \cite{Ch2} may be rephrased in our setting as follows:

\begin{defi}
The NGD condition holds at any time $t\le T$ if there is no $X_{t,T}\in \cR_{t,T}$ such that $\rho_t(X_{t,T})<0$ on a non null set.
\end{defi}

In the setting of Cherny, we suppose that 
\bea \rho_t(X)=\esssup_{Q^t\in \mathcal{D}_t}E_{Q^t}(-X),\label{Cherny}\eea
 where $\mathcal{D}_t$ is a weakly compact subset of $L^1$ with respect to the $\sigma(L^1,L^{\infty})$ topology and we use the definition $E_{Q^t}(-X)=E_{Q^t}(X^-)-E_{Q^t}(X^+)$ with the convention $\infty-\infty=+\infty$. Adapting \cite[Theorem 3.4]{Ch2}, we immediately get the following:

\begin{coro}\label{NGD-NA}
Suppose that the risk-measure is given by (\ref{Cherny}). Then, the NA condition and the NGD condition are equivalent to the existence of a probability measure $Q^t\in \mathcal{D}_t$ such that    the price process $(S_u)_{u=t}^T$ is a $Q^t$-martingale for all $t\le T-1$.

\end{coro}

\section{Appendix: Proofs.}

\noindent Proof of Theorem \ref{na-classical na}. \begin{proof}

 We know that the existence of a risk-neutral probability measure $Q\sim P$ implies AIP. Moreover, suppose that $\rho_t(z\Delta S_{t+1})=0$ on $F_t\in \cF_t$ where  $z\in S(0,1)$. Then, by definition of $\rho_t$, $1_{F_t} z \Delta S_{t+1}\ge 0$. As $E_Q(1_{F_t}z\Delta S_{t+1})=0$, we deduce that $1_{F_t}z\Delta S_{t+1}=0$ hence $\rho_t(-z\Delta S_{t+1})=0$ on $F_t$. By symmetry, we deduce that SRN holds. 

Reciprocally, suppose that  AIP and SRN conditions hold. Let $\theta_t\in L^0(\R^d,\cF_t)$ such that $\theta_t\Delta S_{t+1}\ge 0$ a.s.. Let us write $\theta_t = r_tz_t$ where $r_t\in L^0(\R,\cF_t)$ and $z_t\in L^0(S(0,1),\cF_t)$. On the set $F_t=\{r_t>0\}$, $z_t\Delta S_{t+1}\ge 0$ hence $\essinf_{\cF_{t}}( z_t\Delta S_{t+1})\ge 0$. By the  AIP condition, $\rho_t(z_t\Delta S_{t+1})\ge 0$. We deduce that $\essinf_{\cF_{t}} (z_t\Delta S_{t+1})= 0=\rho_t(z\Delta S_{t+1})$. Under SRN, we deduce that $\rho_t(-z\Delta S_{t+1})=0$ hence $z\Delta S_{t+1}\ge 0$ so that $z_t\Delta S_{t+1}=0$. By a similar reasoning on the set $F_t=\{r_t<0\}$, we also get that $z_t\Delta S_{t+1}=0$ hence $\theta_t\Delta S_{t+1}=0$. We then conclude by \cite[Condition (g), p. 73, Section 2.1.1]{KS}. \end{proof}

\noindent Proof of Theorem \ref{constant_g}. \begin{proof}

If  $\lambda_{t-1}\in L^0(\R,\cF_t)$, $g_{t-1}(\lambda_{t-1} z_{t-1}\Delta S_t)=|\lambda_{t-1}|g_{t-1}(\epsilon_{t-1}z_{t-1}\Delta S_t)$ for some $\epsilon_{t-1}\in L^0(\{-1,1\},\cF_{t-1})$. We deduce that $g_{t-1}(\lambda_{t-1} z_{t-1}\Delta S_t)=0$ on $H_{t-1}$. Recall that $g_{t-1}(\lambda_{t-1} z_{t-1})=\rho_{t-1}(\lambda_{t-1} z_{t-1}\Delta S_t-h_t)$ by Cash invariance. Using the triangular inequality, we then  deduce on $H_{t-1}$ that 
\bean g_{t-1}(0)=\rho_{t-1}(-h_t)&\le& \rho_{t-1}(-\lambda_{t-1} z_{t-1}\Delta S_t)+\rho_{t-1}(\lambda_{t-1} z_{t-1}\Delta S_t-h_t)\\
&\le&g_{t-1}(\lambda_{t-1} z_{t-1}).\eean
Similarly, we have 
\bean g_{t-1}(\lambda_{t-1} z_{t-1})&\le& \rho_{t-1}(\lambda_{t-1} z_{t-1}\Delta S_t)+\rho_{t-1}(-h_t)=\rho_{t-1}(-h_t).
\eean
We deduce that $g_{t-1}(\lambda_{t-1} z_{t-1})=g_{t-1}(0)$ and this implies that $g_{t-1}$ is a constant on the line $\R z_{t-1}$. Indeed, on the contrary  case, the $\cF_{t-1}$-measurable set $\Gamma_{t-1}(\omega)=\{\alpha\in \R:~g_{t-1}(\alpha z_{t-1})\ne g_{t-1}(z_{t-1})\}$ is non empty on the non null set $G_{t-1}=\{\omega\in \Omega:~\Gamma_{t-1}(\omega)\ne \emptyset\}\in \cF_{t-1}$. We then deduce a measurable selection $\tilde z_t\in L^0(\R^d,\cF_{t-1})$ such that $\tilde z_t=\alpha_t z_t$ and $\alpha_t \in \Gamma_{t-1}$ on the set $G_{t-1}$ and we put $\tilde z_t=z_t$ on the complimentary set $\Omega\setminus G_{t-1}$. By the first part above, we deduce that  $g_{t-1}(\tilde z_t)=g_{t-1}( z_t)$ a.s., which contradicts the fact that $\alpha_t\in \Gamma_{t-1}$ on $H_{t-1}$.

\end{proof} \smallskip

\noindent Proof of Proposition \ref{minimum_g}. \begin{proof}
If $\lambda_{t-1}\in L^0(\R,\cF_t)$,  $g_{t-1}(\lambda_{t-1} z_{t-1}\Delta S_t)=|\lambda_{t-1}|g_{t-1}(\epsilon_{t-1}z_{t-1}\Delta S_t)$, where $\epsilon_{t-1}\in L^0(\{-1,1\},\cF_{t-1})$ . Moreover,  $g_{t-1}(\epsilon_{t-1}z_{t-1}\Delta S_t)>0$ on $F_{t-1}$. By sub-additivity, we deduce that $|\lambda_{t-1}|g_{t-1}(\epsilon_{t-1}z_{t-1}\Delta S_t)\le \rho_{t-1}(h_t)+g_{t-1}(\lambda_{t-1} z_{t-1})$. As $|\lambda_{t-1}|$ goes to $+\infty$, we conclude that $g_{t-1}(\lambda_{t-1} z_{t-1})$ tends to $+\infty$ on $F_{t-1}$.

 Now, let us suppose that there is a non null set $G_{t-1}$ of $\cF_{t-1}$ such that $ g_{t-1}(\omega,rz_{t-1})$  does not converge to $+\infty$ if $r\to \infty$ when $\omega\in G_{t-1}$. Note that $\omega\in G_{t-1}$ if and only if there exists $m(\omega)\in \R$ such that, for all $n\ge 1$, there exists $r_n(\omega)\ge n$ such that $g_{t-1}(\omega,r_n(\omega))\le m(\omega)$. Consider the following set
$$\Gamma_{t-1}(\omega)=\{(m,(r_n)_{n=1}^{\infty})\in \R\times \R^{\N}:~r_n\ge n\,\,{\rm and\,}\,g_t(\omega,r_n)\le m,\,\forall n\ge 1\}.$$
The Borel $\sigma$-algebra $\cB(\R^{\N})$ is defined as the smallest topology on $\R^{\N}$ such that the projection mappings $P^n:(r_j)_{j=1}^{\infty}\mapsto r_n$, $n\ge 1$, are continuous. Therefore, we deduce that $\Gamma_{t-1}$ is $\cF_{t-1}$-measurable. As $\Gamma_{t-1}$ is non empty on $G_{t-1}$, we deduce a $\cF_{t-1}$-measurable selection  $(m,(r_n)_{n=1}^{\infty})$ of $\Gamma_{t-1}$ on $G_{t-1}$ that we extend to the whole space $\Omega$ by $m(\omega)=+\infty$ and $r_n(\omega)=n$, if $\omega\in \Omega\setminus G_{t-1}$. Since the $\cF_{t-1}$-measurable sequence $(r_n)_{n=1}^{\infty}$ converges a.s. to $+\infty$, we deduce that $\lim_{n\to+\infty} g_{t-1}(r_nz_{t-1}) = +\infty$ on $G_{t-1}$ by the first part of the proof. This is in contradiction with the property $g_t(\omega,r_n(\omega))\le m(\omega)$, for all $n\ge 1$, if $\omega\in G_{t-1}$.

Similarly, by symmetry, we may also prove that  $\lim_{r\to-\infty} g_{t-1}(rz_{t-1}) = +\infty$ on $F_t$. As $g_{t-1}$ is  l.s.c., we finally deduce that $g_{t-1}$ achieves a minimum on $\R z_{t-1}$.
\end{proof}

\noindent Proof of Theorem \ref{SAIP-price}. \begin{proof}

Note that on the $\cF_{t-1}$-measurable set $\{ \inf_{x\in\R^d}g_{t-1}(x)=+\infty\}$, we have $g_{t-1}(x)=+\infty$ for all $x$ and  $\inf_{x\in\R^d}g_{t-1}(x)=\min_{x\in\R^d}g_{t-1}(x)=+\infty$. The conclusion is then immediate. So, we now suppose w.l.o.g. that there there exists $x_{t-1}\in L^0(\R^d,\cF_{t-1})$ such that $g_{t-1}(x_{t-1})<\infty$. We then introduce the $\cF_{t-1}$-measurable random set $D_{t-1}=\{x\in \R^d:~g_{t-1}(x)\le g_{t-1}(x_{t-1})\}$ and its normalized set defined by   $D^1_{t-1}=\{x/|x|: x\in D_{t-1}\setminus \{0\}\}$.  Observe  that  $\inf_{x\in\R^d}g_{t-1}(x)=\inf_{z\in D_{t-1}}g_{t-1}(z)$. \smallskip

 For any $z\in L^0(D_{t-1}^1\setminus \{0\},\cF_{t-1})$, there exits $r_{t-1}\in L^0((0,\infty),\cF_t)$ such that $r_{t-1}z\in D_{t-1}$.  We have $\rho_{t-1}(z\Delta S_t) > 0$ and $\rho_{t-1}(-z\Delta S_t) > 0$ by  definition of $F_{t-1}$ and $z_{t-1}$.  Notice that, by definition, we have $g_{t-1}(r_{t-1}z) \leq g_{t-1}(x_{t-1}) $.
As  $r_{t-1}> 0$, this is equivalent to
\begin{align*}
     & r_{t-1}\left(zS_{t-1} + \rho_{t-1}\left(zS_t - \dfrac{h_t}{r_{t-1}}\right)\right) \leq  g_{t-1}(x_{t-1}),\\
  &r_{t-1}\left(zS_{t-1} + \rho_{t-1}(zS_t) + \rho_{t-1}\left(zS_t - \dfrac{h_t}{r_{t-1}}\right) - \rho_{t-1}(zS_t)\right) \leq g_{t-1}(x_{t-1}).
\end{align*}
We observe that by convexity and homogeneity:
\begin{align*}
    \rho_{t-1}\left(zS_t - \dfrac{h_t}{r_{t-1}}\right) - \rho_{t-1}(zS_t) \geq -\dfrac{1}{r_{t-1}}\rho_{t-1}(h_t).
\end{align*}
Therefore, $r_{t-1}(zS_{t-1} + \rho_{t-1}(zS_t)) \leq g_{t-1}(x_{t-1}) + \rho_{t-1}(h_t)$, i.e. 
$$ r_{t-1} \leq \dfrac{(g_{t-1}(x_{t-1}) + \rho_{t-1}(h_t))^+}{\rho_{t-1}(z\Delta S_t)}\le \dfrac{(g_{t-1}(x_{t-1}) + \rho_{t-1}(h_t))^+}{\rho_{t-1}(z_{t-1}\Delta S_t)}=M_{t-1}<\infty .$$

This implies that $|x|\le M_{t-1}$ for all $x\in  D_{t-1}$ hence $\inf_{x\in\R^d}g_{t-1}(x)=\inf_{z\in D_{t-1}}g_{t-1}(z)=\inf_{x\in \bar B(0,M_{t-1})}g_{t-1}(x)$ where $\bar B(0,M_{t-1})$ is the closed ball  of radius $M_{t-1}$ and centered at the origin. Since  $\bar B(0,M_{t-1})$ is compact and $g_{t-1}$ is l.s.c., we deduce that $g_{t-1}$ admits a minimum on $\bar B(0,M_{t-1})$ and, finally, the same holds on $\R^d$. By Proposition \ref{measurability}, observe that there exists a measurable version of an argmin, using a measurable selection argument. The conclusion follows.
\end{proof}

\noindent Proof of Theorem \ref{minimal hedging price}. \begin{proof}

Suppose first that $d=2$. Since $\rho_t$ is l.s.c., there exists $z_t\in L^0(S(0,1),\cF_t)$ such that $\inf_{z\in S(0,1)}\rho_t(z\Delta S_{t+1}) = \rho_t(z_t \Delta S_{t+1})$.  By Theorem \ref{SAIP-price} and under SRN,  $g_t$ attains a minimum on $\R^2$ when $\omega\in F_t = \{\rho_t(z_t \Delta S_{t+1})> 0\}\in \cF_t$. 

Let us now suppose that $\omega\in F_t^c =  \{\rho_t(z_t\Delta S_{t+1}) = \rho_t(-z_t\Delta S_{t+1})= 0\}$. We consider a line that is parallel to the line $\R z_t$. For any $z_1, z_2 \in L^0(\R^d,\cF_t)$ on that line such that $z_1 - z_2=r_tz_t \in \R z_t$, $r_t\in L^0(\R,\cF_t)$, we have:
\begin{align*}
    g_t(z_1) &= \rho_t((z_2 + r_tz_t)\Delta S_{t+1} - h_{t+1}) \\
    & \leq \rho_t(z_2\Delta S_{t+1} - h_{t+1}) + \rho_t(r_tz_t \Delta S_{t+1}) = g_t(z_2)
\end{align*}
By symmetry, we also have: $g_t(z_2) \leq g_t(z_1)$, hence $g_t(z_1) = g_t(z_2)$. Therefore,  $g_t$ is constant on any line which is parallel to $\R z_t$. Moreover,
\begin{align*}
    \{(\omega, z_t^{\perp})\in \Omega\times\R^2: z_t^{\perp} z_t(\omega) = 0\}\in \mathcal{F}_t\otimes\mathcal{B}(\R^2).
\end{align*}
By measurable selection argument, 
we may choose $z_t^{\perp}\in L^0(S(0,1),\mathcal{F}_t)$ such that the line $\R z_t^{\perp}$ is orthogonal to $\R z_t$. Since $d=2$,  for any $x\in \R^2$, there exist $\lambda\in\R$ such that $x - \lambda z_t^{\perp} \in \R z_t$. We then deduce from above that:
\begin{align*}
    \inf_{x\in\R^2}g_t(x) =  \inf_{\lambda\in\R}g_t(\lambda z_t^\perp).
\end{align*}

On the set $\{\rho_t(z_t^{\perp}\Delta S_{t+1}) = 0\}$, we get that $\inf_{\lambda\in\R}g_t(\lambda z_t^{\perp}) = g_t(0)$ by Proposition \ref{constant_g}. On the other hand, on the set $\{\rho_t(z_t^{\perp}\Delta S_{t+1}) > 0\}$, we get that $\lim_{|\lambda|\to\infty}g_t(\lambda z_t^{\perp}) = +\infty$ by Proposition \ref{minimum_g} and SRN, hence $g_t$ achieves a minimum on the line $\R z_t^{\perp}$. \smallskip

Let us now prove the $d$-dimensional case by induction. Recall that there exists $z_t\in L^0(S(0,1),\mathcal{F}_t)$ such that  $\rho_t(z_t\Delta S_{t+1})= \inf_{z\in S(0,1)}\rho_t(z\Delta S_{t+1})$. On $F_t = \{\rho_t(z_t\Delta S_{t+1}) >0\}$, by Theorem \ref{SAIP-price} and SRN, $g_t$ attains a minimum on $\R^d$. On $F_t^c = \{\rho_t(z_t\Delta S_{t+1}) =0\}$, consider a hyperplane $I_{d-1}$ which is orthogonal to $\R z_t$ and admits an orthonormal basis $(z_1, z_2,...,z_{d-1})$ such that for each $\omega\in \Omega$, $\hat z=(z_t, z_1,..., z_{d-1})$ is an orthonormal basis for $\R^d$. Note that each $z_i$ can be chosen  in $L^0(S(0,1),\mathcal{F}_t)$. Indeed, similarly to the case $d=2$, we first choose $z_1 \in L^0(S(0,1),\mathcal{F}_t)$ orthogonal to $z_t$. Recursively, for $i\in\{2,...,d-1\}$, we have:
\begin{align*}
    \{(\omega,z_i)\in\Omega\times\R^d: z_iz_j(\omega) = 0 \text{ for all } j = 0,...,i-1\}\in\mathcal{F}_t\otimes\mathcal{B}(\R^d).
\end{align*}
By measurable selection argument, we  then choose $z_i\in L^0(S(0,1),\mathcal{F}_t)$.
We denote by $M_{t}$ the matrix such that $\hat z_{i}=M_t e_i$, for every $i\ge 1$, where $e_i=(0,\cdots,1,\cdots,0)\in \R^d$. We recall the change of variable $x = M_t\tilde{x}$ where $x$ and $\tilde x$ are the coordinates of an arbitrary vector of $\R^d$ in the basis $(e_i)_{i\ge 1}$ and $(\hat z_i)_{i\ge 1}$ respectively. The $i$th column vector of $M_{t}$ coincides with $\hat z_i$ expressed in the  basis $(e_i)_{i\ge 1}$, hence each entry of $M_t$ belongs to $L^0(\R,\mathcal{F}_t)$ and  so do the components of $M_t^{-1}$.
We then define the adapted processes $\tilde{S}_{u}=M_t^{-1}S_{u}=M_t'S_u$, for $u=t,t+1$. We have:
\begin{align*}
    g_t(x) = \rho_t(x\Delta S_{t+1} -h_{t+1}) = \rho_t(\tilde{x}\Delta \tilde{S}_{t+1} - h_{t+1}).
\end{align*}
We observe that $\tilde{S}_{u=t,t+1}$ forms a new market model which also satisfies the  NA condition between $t$ and $t+1$. Indeed, for any $z\in S(0,1)$, we have:
\begin{align*}
    \rho_{t}(z\Delta\tilde{S}_{t+1}) = \rho_{t}(zM_t'\Delta S_{t+1}),
\end{align*}
hence $\rho_{t}(z\Delta\tilde{S}_{t+1}) = 0$ implies that $\rho_{t}(-zM_t'\Delta S_{t+1}) = 0$ by the NA condition satisfied in the market formed by $S$ which, in turn, implies $\rho_t(-z\Delta\tilde{S}_{t+1}) = 0$. \smallskip

Fix $\omega$ and, for any $x\in\R^d$, consider the orthogonal projection $\bar{x}$ of $x$ onto $I_{d-1}$. We then have $g_t(x) = g_t(\bar{x})$. For $\bar{x}\in I_{d-1}$, we denote $\hat{x} = M_t^{-1}\bar{x}$, we have:
\begin{align*}
    \bar{x}\Delta S_{t+1} =  \hat x\Delta\tilde{S}_{t+1}: = \sum_{i=1}^d\hat{x}^i\Delta\tilde{S}_{t+1}^i = \sum_{i=2}^d\hat{x}^i\Delta\tilde{S}_{t+1}^i,
\end{align*}
since the first coordinate of $\hat{x}$ equals $0$ in the new basis. We deduce that:
\begin{align*}
    \inf_{x\in\R^d}g_t(x) =  \inf_{x\in I_{d-1}}\rho_t(x\Delta S_{t+1} - h_{t+1}) 
    &= \inf_{\hat{x}\in\R^{d-1}}\rho_t\left(\sum_{i=2}^d\hat{x}^i\Delta\tilde{S}_{t+1}^i - h_{t+1}\right)
\end{align*}
This means that we have reduced the optimization problem to a market with only $d-1$ assets defined by $(\tilde{S}^2,...,\tilde{S}^d)$. As it satisfies the NA condition, we deduce that $ \inf_{x\in\R^d}g_t(x)$ is attained by induction.\end{proof}

\noindent Proof of Theorem \ref{SAIP}. \begin{proof}
Suppose that  NA holds. By Theorem \ref{minimal hedging price}, there is  $z_t\in L^0(S(0,1),\mathcal{F}_t)$ and $r_t\in L^0(\R,\mathcal{F}_t)$ such that $P^*_t =  \rho_t(r_tz_t\Delta S_{t+1} - h_{t+1})$.
 Suppose that  $\rho_{t}(z_t \Delta S_{t+1})$ and $\rho_{t}(-z_t\Delta S_{t+1})$ are both equal to $0$. Then, the function $g_t$ associated to $h_{t+1}$, see  (\ref{gt-1}), is  constant on the line $\R z_t$ by Theorem \ref{constant_g}. Therefore, $P_t^* = g_t(0) = \rho_t(-h_{t+1}) > 0$. Otherwise, under NA, $\rho_{t}(z_t\Delta S_{t+1}) > 0$ and $\rho_{t}(-z_t\Delta S_{t+1}) > 0$. Using triangular inequalities,  and the assumption $\rho_t(h_{t+1})\le 0$, we then deduce that:
\bean
     P^*_t &=& r_tz_tS_t + \rho_t(r_tz_tS_{t+1} - h_{t+1}), \\
     &=& \rho_{t}(-h_{t+1})1_{\{r_t=0\}} + r_t \rho_t\left(z_t\Delta S_{t+1} - \dfrac{h_{t+1}}{r_t}\right)1_{\{r_t>0\}}\\
     &\;& - r_t \rho_t\left(-z_t\Delta S_{t+1} + \dfrac{h_{t+1}}{r_t}\right)1_{\{r_t<0\}},\\
     &\geq& \rho_{t}(-h_{t+1})1_{\{r_t=0\}} + r_t \rho_t\left(z_t\Delta S_{t+1}\right)1_{\{r_t>0\}} - r_t \rho_t\left(-z_t\Delta S_{t+1}\right)1_{\{r_t<0\}},  \\
     & > & 0.
\eean

For the reverse implication, let us  prove first that AIP holds. We fix $h_{t+1}$ such that $\rho_{t}(-h_{t+1})\in (0,\infty)$ and $\rho_{t}(h_{t+1})\le 0$. So, with the function $g_t$ associated to $h_{t+1}$, see  (\ref{gt-1}), we have $P_{t}^*=P_{t}^*(h_{t+1})=\inf\limits_{x\in \R} g_{t}(x)>0$ by assumption and $g_{t}(rz)>0$ for all $r\in \R$ and  $z\in S(0,1)$. Let us show that the set $\{zS_{t}+\rho_{t}(zS_{t+1})<0\}$ is empty  for all $z\in S(0,1)$ a.s.. In the contrary case, by measurable selection, we may construct $z_t\in L^0(\R^d,\cF_t)$ such that we have $\bP(z_tS_{t}+\rho_{t}(z_tS_{t+1})<0)>0$. We then define $$r_{t}:=-\frac{\rho_{t}(-h_{t+1})}{\rho_{t}(z_t\Delta S_{t+1})}1_{\{\rho_{t}(z_t\Delta S_{t+1})<0\}}\ge 0.$$
We have
\bean g_{t}(r_tz_t)&=&r_tz_tS_{t}+\rho_{t}(r_tz_tS_{t+1}-h_{t+1}),\\
& \le&r_tz_tS_{t}+\rho_{t}(r_tz_tS_{t+1})+\rho_{t}(-h_{t+1}),\\
&\le& r_{t}\rho_{t}(z_t\Delta S_{t+1})+\rho_{t}(-h_{t+1}),\\
&\le&\rho_{t}(-h_{t+1})1_{\{\rho_{t}(z_t \Delta S_{t+1})\ge 0\}.}\eean
Therefore, $P_{t}^*\le 0$ on the set $\{\rho_{t}(z_t\Delta S_{t+1})<0\}$ in contradiction with $P_{t}^*>0$. \smallskip

Let us  show that $\rho_{t}(-z\Delta S_{t+1})=0$ if $\rho_{t}(z\Delta S_{t+1})=0$ for any $z\in S(0,1)$. Otherwise, by measurable selection argument,  there exists $z_t\in L^0(S(0,1),\cF_t)$ such that $\Lambda_{t}:=\{\rho_{t}(z_t\Delta S_{t+1})=0\}\cap \{\rho_{t}(-z_t\Delta S_{t+1}) > 0\}$ satisfies $\bP(\Lambda_{t})>0$. If $h_{t+1}= z_t\Delta S_{t+1}$, then  $\rho_{t}(-h_{t+1})= \rho_{t}(-z_t\Delta S_{t+1})>0$ on $\Lambda_{t}$. On the complimentary set, we fix $h_{t+1}=\gamma_{t}>0$, $\gamma_{t}\in L^0((0,\infty),\cF_{t})$. It follows that $\rho_{t}(-h_{t+1})>0$. Moreover, $\rho_{t}(h_{t+1})=\rho_{t}(z_t\Delta S_{t+1})=0$ on $\Lambda_{t}$ and, otherwise, $\rho_{t}(h_{t+1})=-\gamma_{t}<0$. Therefore, $\rho_{t}(h_{t+1})\le 0$. We deduce that   $P_{t}^*(h_{t+1})>0$, by assumption. On the other hand, if $r\ge 1$, and $\omega\in \Lambda_{t}$,
$$P_{t}^*(h_{t+1})\le \rho_{t}(rz_t\Delta S_{t+1}- z_t\Delta S_{t+1})=(r-1)\rho_{t}(z_t\Delta S_{t+1})=0.$$ It follows that  $P_{t}^*(h_{t+1})\le 0$ on $\Lambda_{t}$, i.e. a contradiction. We conclude that $\rho_{t}(z\Delta S_{t+1})=0$ if and only if $\rho_{t}(-z\Delta S_{t+1})=0$ for any $z\in S(0,1)$. 

At last, it is clear that $P_{t}^*(h_{t+1})\le g_t(0)=\rho_{t}(-h_{t+1})$. Moreover, for all $x\in \R^d$, $0\le \rho_{t}( x\Delta S_{t+1})\le \rho_{t}( x\Delta S_{t+1}-h_{t+1})+\rho_{t}(h_{t+1})$. Taking the infimum in the r.h.s. of this inequality, we get that $0\le P_{t}^*(h_{t+1})+\rho_{t}(h_{t+1})$ and we may conclude. \end{proof}

\noindent Proof of Theorem \ref{propRP0}. \begin{proof}

 By  \cite{Acc}, \cite{Del2}, there exists $\mathcal{D}_t$ such that (\ref{RP}) holds if $X\in L^{\infty}$. By homogeneity, it is clear that (\ref{RP}) still holds if $X$ is $\cF_t$-bounded, i.e. $|X|\le c_t$ where $c_t\in L^0(\R_+,\cF_t)$.  Let us show that  (\ref{RP}) still holds for any random variable $X$ such that $X\le c_t$ a.s. for some $c_t\in L^0(\R_+,\cF_t)$. Let us first suppose that  $X$ is acceptable. Let us define $X^M=X1_{\{X\ge -M\}}$ for any $M>0$. Then, $X^M$ is $\cF_t$-bounded a.s.. As $X^M=X-X1_{\{X<-M\}}$ and $-X1_{\{X<-M\}}\ge 0$, then $X^M$ is acceptable i.e. $\rho_t(X^M)\le 0$. By (\ref{RP}), we deduce that $E_{Q}(X^M|\cF_t)\ge 0$ for all $Q\in \mathcal{D}_t$. Thus, $E_{Q}((X^M)^+|\cF_t)\ge E_{Q}((X^M)^-|\cF_t)$ and, as $M\to \infty$, we get that $c_t\ge E_Q(X^+|\cF_t)\ge E_Q(X^-|\cF_t)$ hence $\infty>E_Q(X|\cF_t)\ge 0$. More generally, for any $X$ such that $X\le c_t$ for some  $c_t\in L^0(\R_+,\cF_t)$, $\rho_t(X)+X$ is acceptable hence $\rho_t(X)\ge E_Q(-X|\cF_t)$ for any $Q\in \mathcal{D}_t$. We deduce that the inequality $\rho_t(X)\ge \esssup_{ Q\in \mathcal{D}_t}E_Q(-X|\cF_t)$ holds. 

For the reverse inequality, note that the random variable
$$\gamma^M=\esssup_{ Q\in \mathcal{D}_t}E_Q(-X|\cF_t)+X^M\in [-c_t+X^M,X^M]$$
is $\cF_t$-bounded hence (\ref{RP})  holds for $\gamma^M$, as seen above. Moreover, we have $E_Q(-\gamma^M|\cF_t)\le E_Q(X|\cF_t)-X^M=E_Q(X1_{X<-M}|\cF_t)\le 0$. We deduce by (\ref{RP}) that $\rho_{t-1}(\gamma^M)\le 0$. Using the Cash invariance property, we deduce that $\rho_{t-1}(X^M)\le
\esssup_{ Q\in \mathcal{D}_t}E_Q(-X|\cF_t)$. As $\lim_{M\to \infty}X^M=X$, we then deduce that $\rho_{t-1}(X)\le \liminf_{M\to \infty}\rho_{t-1}(X^M)\le
\esssup_{ Q\in \mathcal{D}_t}E_Q(-X|\cF_t)$ so that we may conclude that the equality (\ref{RP}) holds for any random variable that are $\cF_t$-bounded form above.
\end{proof}

 \noindent Proof of Theorem \ref{DR}. \begin{proof}

 Suppose that $Z=X+\epsilon^+$ where $X$ is $\cF_t$-bounded from above and acceptable and $\epsilon^+\ge 0$ a.s..  Then, $\cD_t$ exists by Proposition \ref{propRP0} and, for all $Q\in \cD_t$, $E_Q(Z|\cF_t)\ge E_Q(X|\cF_t)\ge 0$. As $\rho_t(Z)+Z$ admits the same form than $Z$, we deduce that $\rho_t(Z)+Z$ admits non negative conditional expectations under $Q\in \cD_t$. Therefore, $\rho_t(Z)\ge E_Q(-Z|\cF_t)$ for all $Z\in \cD_t$ hence $\rho_t(Z)\ge \esssup_{Q\in \cD_t}E_Q(-Z|\cF_t)$, at least when $\rho_t(Z)>-\infty$.
Otherwise, when $\rho_t(Z)=-\infty$, $Z_{\alpha}=-\alpha+Z$ is acceptable for all $\alpha>0$, hence $E_Q(Z_{\alpha}|\cF_t)\ge 0$, i.e. $E_Q(Z|\cF_t)\ge \alpha$ for all $\alpha>0$. It follows that  $E_Q(Z^-|\cF_t)-E_Q(Z^+|\cF_t)\le -\alpha$ and finally, as $\alpha\to \infty$, we deduce that $\rho_t(Z)=\esssup_{Q\in \cD_t}E_Q(-Z|\cF_t)=-\infty$.

 Consider an acceptable position $Z$. Then, by assumption, $Z=\limsup_n Z^n$ where $Z^n$ is of the form $Z^n=X^n+\epsilon_n^+$ with $\epsilon_n^+\ge 0$ a.s. and $X^n$ is $\cF_t$-bounded from above. Note that $\sup_{k\le n\le m}X_n$ is still $\cF_t$-bounded from above for all $m\ge k \ge  1$. Since  $\sup_{n\ge k}Z_n\ge  \sup_{k\le n\le m}Z_n\ge  \sup_{k\le n\le m}X_n$, for all $m\ge k$, we deduce that  $\sup_{n\ge k}Z_n$ is of the form $X_k+\epsilon^+_k$ where  $X_k$ is $\cF_t$-bounded from above and acceptable while $\epsilon^+_k\ge 0$  a.s.. It follows that any acceptable position is of the form $Z=\lim\downarrow Z_n$ where $Z_n$ is of the form $Z_n=X_n+\epsilon^+_n$ and  $X_n$ is $\cF_t$-bounded from above and acceptable while $\epsilon^+_n\ge 0$  a.s..  As $Z\le Z_n$, we deduce that $\rho_t(Z)\ge \rho_t(Z_n)\ge \esssup_{Q\in \cD_t}E_Q(-Z_n|\cF_t)$ by virtue of the inequality we have shown in the first part. As $(-Z_n)$ is non decreasing we finally deduce that $\rho_t(Z)\ge E_Q(-Z|\cF_t)$ for any $Q\in \cD_t$, when $n\to \infty$. It follows that $\rho_t(Z)\ge \esssup_{Q\in \cD_t}E_Q(-Z|\cF_t)$. 
 
 Moreover, suppose that (\ref{RP1}) holds for any acceptable position $Z_n$ of the form $Z_n=X_n+\epsilon_n^+$ where $X_n$ is $\cF_t$-bounded from above and acceptable and $\epsilon_n^+\ge 0$ a.s.. By lower semi-continuity, 
 $$\rho_t(Z)\le \liminf_n \rho_t(Z_n)=\liminf_n\esssup_{Q\in \cD_t}E_Q(-Z_n|\cF_t).$$ As $Z\le Z_n$, $E_Q(-Z_n|\cF_t)\le E_Q(-Z|\cF_t)$, and we deduce the inequality $\rho_t(Z)\le \esssup_{Q\in \cD_t}E_Q(-Z|\cF_t)$. We then conclude that (\ref{RP1}) holds for every acceptable position $Z$ and, finally, for every $X\in L^0$ as $\rho_t(X)+X$ is acceptable.

 It remains to show that (\ref{RP1}) holds for $Z=X+\epsilon^+\in \mathcal A_t^{\infty,+}+L^0(\R^+,\cF_T)$. To get it, it is sufficient to prove that $\rho_t(Z)\le \esssup_{Q\in \cD_t}E_Q(-Z|\cF_t)$. Let us define $Z^n=X+\epsilon^+1_{\{\epsilon^+\le n\}}+\alpha_n 1_{\{\epsilon^+> n\}}\in \mathcal A_t^{\infty,+}$ where $\alpha_n>0$ is chosen large enough in such a way that $P(\alpha_n<\epsilon^+)<\varepsilon$. Then, $(\alpha_n-\epsilon^+)1_{\{\epsilon^+>n\}}$ is acceptable by hypothesis for $P((\alpha_n-\epsilon^+)1_{\{\epsilon^+>n\}}<0)\le P(\alpha_n<\epsilon^+)<\varepsilon$. 
 
Since $Z^n\to Z$ a.s., we deduce that $\rho_t(Z)\le \liminf_n\rho_t(Z^n)$. Recall that $\rho_t(Z^n)=\sup_{Q\in \cD_t}E_Q(-Z^n|\cF_t)$ by Proposition \ref{propRP0}. Hence, 
$$\rho_t(Z^n)\le \esssup_{Q\in \cD_t}E_Q(-Z|\cF_t)+\esssup_{Q\in \cD_t}E_Q(Z-Z^n|\cF_t).$$

Moreover, since $Z^n-Z$ is $\cF_t$-bounded from above, we have 
$$\esssup_{Q\in \cD_t}E_Q(Z-Z^n|\cF_t)=\rho_t(Z^n-Z)=\rho_t((\alpha_n-\epsilon^+)1_{\{\epsilon^+>n}\})\le 0.$$ We then deduce that  $\rho_t(Z)\le\ \esssup_{Q\in \cD_t}E_Q(-Z|\cF_t)$ and the conclusion follows. \end{proof}

 \noindent Proof of Theorem \ref{ClosureNA}. \begin{proof}

 Consider $\theta_t\in L^0(\R^d,\cF_t)$. By Theorem \ref{AIP-theo}, it suffices to show that $\rho_t(\theta_t \Delta S_{t+1})\ge 0$ a.s.. Otherwise, the set $\Lambda_t=\{\rho_t(\theta_t \Delta S_{t+1})< 0\}$ admits a positive probability and $\theta_t \Delta S_{t+1}1_{\Lambda_t}\in \cR_{t,T}\cap \cA_{t,T}=\cA_{t,T}^0$. It follows that $\rho_t(\theta_t \Delta S_{t+1}1_{\Lambda_t})=0$ hence a contradiction. Therefore, AIP holds. 

Let us  show that $\overline{\cZ_{t,T}}\subseteq \cZ_{t,T}$.  In the one step model, let us suppose that $\gamma^n=\theta_{T-1}^n\Delta S_T-\epsilon_{T-1,T}^{n}\in \cZ_{T-1,T}$ converges  to $\gamma^{\infty}\in L^0(\R,\cF_{T})$ in probability. We suppose that $\epsilon_{T-1,T}^{n}\in \cA_{T-1,T}$.  We need to show that $\gamma^{\infty}\in \cZ_{T-1,T}$.

On the $\cF_{T-1}$-measurable set $\Lambda_{T-1}:=\{\liminf_n|\theta_{T-1}^n|<\infty\}$, by \cite[Lemma 2.1.2]{KS}, we may assume w.l.o.g. that $\theta_{T-1}^n$ is convergent to some $\theta_{T-1}^{\infty}$ hence $\epsilon_{T-1,T}^{n}$ is also convergent and, finally, $\gamma^{\infty}1_{\Lambda_{T-1}}\in \cZ_{T-1,T}$.

Otherwise, on $\Omega\setminus \Lambda_{T-1}$, we introduce the  sequences,
$$\tilde \theta_{T-1}^{n}:=\theta_{T-1}^{n}/(|\theta_{T-1}^n|+1),  \;\; \tilde\epsilon_{T-1,T}^{n}:=\epsilon_{T-1,T}^{n}/(|\theta_{T-1}^n|+1).$$ By \cite[Lemma 2.1.2]{KS}, we may assume that a.s. $\tilde \theta_{T-1}^n\to \tilde \theta_{T-1}^{\infty}$, $\tilde \epsilon_{T-1,T}^{n}\to \tilde \epsilon_{T-1,T}^{\infty}$ and
$\tilde \theta_{T-1}^{\infty}\Delta S_T-\tilde \epsilon_{T-1,T}^{\infty}=0 \mbox{ a.s.}.$
Note that $|\tilde \theta_{T-1}^{\infty}|=1$ a.s.. As $\tilde \theta_{T-1}^{\infty}\Delta S_T$ is acceptable ($\epsilon_{T-1,T}^{\infty}\in \mathcal{A}_{T-1,T}$) then $\tilde \theta_{T-1}^{\infty}\Delta S_T\in \cA_{t,T}^0$ by assumption.   We follow the  recursive arguments on the dimension of \cite{KSt}. Since $|\tilde \theta_{T-1}^{\infty}|=1$, there exists a partition of $\Omega\setminus\Lambda_{T-1}$ into $d$ disjoint subsets $G_{T-1}^i\in\mathcal{F}_{T-1}$ such that $\tilde \theta_{T-1}^{\infty,i} \neq 0$ on $G_{T-1}^i$. Define on $G_{T-1}^i$, $\hat{\theta}^{n,i}_{T-1}:= \theta^n_{T-1} - \beta_{T-1}^{n,i} \tilde \theta_{T-1}^{\infty}$ where $\beta_{T-1}^{n,i}:= \theta^{n,i}_{T-1}/\tilde \theta_{T-1}^{\infty,i}$. Observe that
$\gamma^n=\hat{\theta}^{n,i}_{T-1}\Delta S_T-\tilde\epsilon_{T-1,T}^{n,i}$ where the position  $\tilde\epsilon_{T-1,T}^{n,i}=\epsilon_{T-1,T}^{n}-\beta_{T-1}^{n,i} \tilde \theta_{T-1}^{\infty}\Delta S_T$ is acceptable since $\pm \tilde \theta_{T-1}^{\infty}\Delta S_T$ are acceptable. As  $\hat \theta_{T-1}^{n,i} = 0$ on $G_{T-1}^i$, we repeat the entire procedure on each $G^i_{T-1}$ with the new expression $\gamma^n=\hat{\theta}^{n,i}_{T-1}\Delta S_T-\tilde\epsilon_{T-1,T}^{n,i}$ such that the number of components of $\hat{\theta}^{n,i}_{T-1}$ is reduced by one. We then conclude by recursion on the number of non-zero components since the conclusion is trivial if all the coordinates vanish. \smallskip

We now show the result in the multi-step models by induction. Fix some $s \in \{t,\ldots,T-1\}$. We show that $\overline{\cZ_{s+1,T}} \subseteq \cZ_{s+1,T}$ implies the same property for $s$ instead of $s+1$. \smallskip

Since AIP holds, we get that $\cZ_{s+1,T} \cap L^0(\R_+,\cF_{s+1})=\{0\}$ hence, from $\overline{\cZ_{s+1,T}} \subseteq \cZ_{s+1,T}$, we get   that $\overline{\cZ_{s+1,T}} \cap L^1(\R_+,\cF_{s+1})=\{0\}.$ Using the Hahn-Banach separation theorem in $L^1$, we deduce the existence of $Q^{(s+1)}\ll P$ with $\frac{dQ^{(s+1)}}{dP}\in L^{\infty}$ such that $\rho_{s+1}:=E_P(\frac{dQ^{(s+1)}}{dP}|\cF_{s+1})=1$ a.s., $(S_u)_{u\ge s+1}$ is a martingale under $Q^{(s+1)}$ and $E_Q(a_{s+1,T}|\cF_{s+1})\ge 0$ for all $a_{s+1,T}\in \cA_{s+1,T}$ such that $E_Q(|a_{s+1,T}||\cF_{s+1})<\infty$ a.s.. Suppose that
$$\gamma^n=\sum_{u=s+1}^{T} \theta_{u-1}^n\Delta S_{u}-\epsilon_{s,T}^{n} \in \cZ_{s,T} \mbox{ converges to } \gamma^{\infty}\in L^0(\R,\cF_{T}).$$
We suppose that $\epsilon_{s,T}^{n}\in \cA_{s,T}$. By Lemma \ref{TC-acc}, $\epsilon_{s,T}^{n}=\epsilon_{s,s+1}^{n}+\epsilon_{s+1,T}^{n}$, where $\epsilon_{s,s+1}^{n}\in \cA_{s,s+1}$ and $\epsilon_{s+1,T}^{n}\in \cA_{s+1,T}$.  As before, on the $\cF_{s}$-measurable set $\Lambda_{s}:=\{\liminf_n|\theta_{s}^n|<\infty\}$,  we may assume w.l.o.g. that $\theta_{s}^n$  converges to $\theta_{s}^{\infty}$. Therefore, on $\Lambda_{s}$,
$$\sum_{u=s+2}^{T} \theta_{u-1}^n\Delta S_{u}-\epsilon_{s,T}^{n}=\gamma^n - \theta_{s}^n\Delta S_{s+1} \to
\gamma^{\infty} - \theta_{s}^{\infty} \Delta S_{s+1}.$$

On the subset $\hat \Lambda_{s+1}:=\{\liminf_n|\epsilon_{s,s+1}^n|=\infty\}\cap\Lambda_{s}\in \cF_{s+1}$, we use the normalization procedure as previously, i.e. we divide by $|\epsilon_{s,s+1}^n|$, up to a subsequence,  and,  by the induction hypothesis,  we obtain that 

$$\sum_{u=s+2}^{T} \tilde\theta_{u-1}^n\Delta S_{u}-\tilde\epsilon_{s+1,T}=\tilde\epsilon_{s,s+1},$$
where $\tilde\epsilon_{s+1,T}\in \cA_{s+1,T}$ and $\tilde\epsilon_{s,s+1}\in \cA_{s,s+1}$ satisfies $|\tilde\epsilon_{s,s+1}|=1$ a.s.. Moreover, since $S$ is a martingale,
\bean E_{Q^{(s+1)}}\left(\sum_{u=s+2}^{T} \tilde\theta_{u-1}^n\Delta S_{u}|\cF_{s+1}\right)=0. 
\eean
Moreover, still by assumption,  $E_{Q^{(s+1)}}(\tilde\epsilon_{s+1,T}|\cF_{s+1})\ge 0$. We deduce that $\tilde\epsilon_{s,s+1}=E_{Q^{(s+1)}}(\tilde\epsilon_{s,s+1}|\cF_{s+1})\le 0$. Therefore, $\tilde\epsilon_{s,s+1}=-1$ hence $\rho_s(\tilde\epsilon_{s,s+1})=\rho_s(-1)=1$, which is in contradiction with $\rho_s(\tilde\epsilon_{s,s+1})\le 0$. Therefore, we may suppose, on $\Lambda_{s}$, that $\epsilon_{s,s+1}^n$ converges a.s. to some $\epsilon_{s,s+1}\in \cA_{s,s+1}$. By the induction hypothesis, we then deduce that  $\sum_{u=s+2}^{T} \theta_{u-1}^n\Delta S_{u}-\epsilon_{s+1,T}^{n}$ also converges to an element of $\cZ_{s+1,T}$ and we  conclude that $\gamma^{\infty} 1_{\Lambda_{s}} \in \cZ_{s,T}$.\\
On $\Omega\setminus \Lambda_{s}$, we use the normalisation procedure as before, and deduce the equality
$$\sum_{u=s+1}^{T} \tilde \theta_{u-1}^{\infty}\Delta S_{u}-\tilde \epsilon_{s,T}^{{\infty}}=0 \mbox{ a.s.}$$
for some $\tilde \theta_{u}^{\infty}\in L^0(\R,\cF_{u})$, $u \in \{s,\ldots,T-1\}$ and $\tilde \epsilon_{s,T}^{{\infty}}\in \cA_{s,T}$.  By Lemma \ref{TC-acc}, we write $\tilde \epsilon_{s,T}^{{\infty}}=\tilde \epsilon_{s,s+1}^{{\infty}} +\tilde \epsilon_{s+1,T}^{{\infty}}$ where $\tilde \epsilon_{s,s+1}^{{\infty}}\in \cA_{s,s+1}$ and $\tilde \epsilon_{s+1,T}^{{\infty}}\in \cA_{s+1,T}$. Moreover,  $|\tilde \theta_{s}^{\infty}|=1$ a.s.. We deduce that:

$$\tilde{\theta}_s^\infty\Delta S_{s+1}+\sum_{u=s+2}^{T} \tilde \theta_{u-1}^{\infty}\Delta S_{u}-\tilde \epsilon_{s+1,T}^{{\infty}}=\tilde \epsilon_{s,s+1}^{{\infty}}\mbox{ a.s..}$$ Taking the conditional expectation knowing $\cF_{s+1}$ under $Q^{(t+1)}$, we deduce that $\tilde \epsilon_{s,s+1}^{{\infty}}\le \tilde{\theta}_s^\infty\Delta S_{s+1}$. It follows that $\rho_s(\tilde{\theta}_s^\infty\Delta S_{s+1})\le \rho_s(\tilde \epsilon_{s,s+1}^{{\infty}})\le 0$ hence $\tilde{\theta}_s^\infty\Delta S_{s+1}\in \cA_{s,T}^0$ by the assumption. Using the  one step arguments  based on the elimination of non-zero components of the sequence $\theta_s^n$, we may replace $\theta_s^n$ by $\tilde{\theta}_s^n$ such that $\tilde{\theta}_s^n$ converges. We then repeat the same arguments on the set $\Lambda_s$ to conclude that $\gamma^{\infty} 1_{\Omega\setminus\Lambda_{s}} \in  \cZ_{s,T}$. \end{proof}

 \noindent Proof of Theorem \ref{NA-Formulation}. \begin{proof}

 Let us consider $W_{t,T}\in \cR_{t,T}\cap \cA_{t,T}$ Then, $W_{t,T}$ is of the form:
$$W_{t,T}=\sum_{s=t+1}^T \theta_{s-1}\Delta S_s=\sum_{s=t+1}^T a_{s-1,s},$$
where $\theta_{s-1} \in L^0(\R,\cF_{s-1})$ and  $a_{s-1,s}\in \cA_{s-1,s}$, for all $s=t+1,\cdots,T$. It follows that: 
\bea \label{aux-NA-Formulation}\theta_{t}\Delta S_{t+1}-a_{t,t+1}+\sum_{s=t+2}^T \left(\theta_{s-1}\Delta S_s-a_{s-1,s}\right)=0.\eea

Therefore, $p_t=\theta_{t}\Delta S_{t+1}-a_{t,t+1}$ is a (direct) price at time $s=t+1$ for the zero claim. Under AIP condition, we get that $\theta_{t}\Delta S_{t+1}\ge a_{t,t+1}$ hence $ \rho_{t}(\theta_{t}\Delta S_{t+1})\le 0$. As $ \rho_{t}(\theta_{t}\Delta S_{t+1})\ge 0$ by AIP, $ \rho_{t}(\theta_{t}\Delta S_{t+1})= 0$ and, by SRN, we get that $\rho_{t}(\theta_{t}\Delta S_{t+1}) = \rho_{t}(-\theta_{t}\Delta S_{t+1})=0$. We then deduce that  $-p_t\in 
\cA_{t,T}\cap L^0(\R_-,\cF_T)=\{0\}$ hence $p_t=0$ and $\theta_{t}\Delta S_{t+1}=a_{t,t+1}\in \cA_{t,T}^0$. The equality (\ref{aux-NA-Formulation}) may be rewritten as:

\bea \label{aux1-NA-Formulation}\theta_{t+1}\Delta S_{t+2}-a_{t+1,t+2}+\sum_{s=t+3}^T \left(\theta_{s-1}\Delta S_s-a_{s-1,s}\right)=0.\eea
By induction, we finally deduce that  $\theta_{s}\Delta S_{t+1}=a_{s,s+1}\in \cA_{s,s+1}^0$ for all $s\ge t$. By Remark \ref{remark acceptance set}, we have $W_{t,T}\in\cA_{t,T}^0$.

Consider now  $\epsilon^+_T\in \cZ_{t,T}\cap L^0(\R_+,\cF_T)$. We may write $\epsilon^+_T=r_{t,T}-a_{t,T}$ where $r_{t,T}\in \cR_{t,T}$ and $a_{t,T}\in \cA_{t,T}$. We get that $r_{t,T}=a_{t,T}+\epsilon^+_T\in \cR_{t,T}\cap \cA_{t,T}=\cA_{t,T}^0$ hence $-r_{t,T}\in \cA_{t,T}$. It follows that $-\epsilon^+_T\in \cA_{t,T}\cap L^0(\R_-,\cF_T)=\{0\}$. \end{proof}

 \noindent Proof of Theorem \ref{FTAP}. \begin{proof}

Suppose that 1) holds. By Theorem \ref{NA-Formulation}, we deduce that 3) holds. Note that 2) and 3) are equivalent since the risk measure is time-consistent.  Suppose that 3) holds. Since $-\cA_{t,T}\subseteq\cZ_{t,T}$, it follows that $\cA_{t,T}^0\subseteq\cZ_{t,T}\cap\cA_{t,T}$. Reciprocally, consider $x_{t,T} = W_{t,T} - a_{t,T}\in \cZ_{t,T}\cap\cA_{t,T}$, where $W_{t,T} \in\cR_{t,T}$ and $a_{t,T}\in\cA_{t,T}$, then $W_{t,T}\in\cA_{t,T}$ hence $W_{t,T}\in\cA_{t,T}^0$ by 2).  It follows that  $x_{t,T}\in(-\cA_{t,T})$ and we conclude that  $\cZ_{t,T}\cap\cA_{t,T}= \cA_{t,T}^0$. Moreover, by Theorem \ref{ClosureNA}, $\cZ_{t,T}$ is closed in probability hence 4) holds. Note that 4) and 5) are equivalent since the risk measure is time-consistent.\smallskip

Assume that 4) holds. The existence of $Q$ in 7) holds by standard arguments based on the Hahn-Banach separation theorem.  In particular, NA holds under $P'$ such that $P'\sim P$. We suppose w.l.o.g that $S_t$ is integrable under $P$ for every $t$. If $x\in L^1(\R,\cF_T)\cap(\cA_{t,T}\setminus\cA_{t,T}^0)$, $x\notin \cZ_{t,T}\cap L^1(\R,\cF_{T})$. By the Hahn-Banach separation theorem, there exists $p_x\in L^\infty(\R,\cF_T)$ and $c\in\R$ such that $E(p_x X) < c < E(xp_x),\, \forall X\in\cZ_{t,T}.$
As $\cZ_{t,T}$ is a cone, we get that $E(p_xX) \leq 0$ for all $X\in\cZ_{t,T}$ and since $-L^0(\R_+,\cF_T)\subseteq\cZ_{t,T}$, we deduce that $p_x\geq 0$ a.s.. With $X = 0$, we get that $E(xp_x) > 0$ and, as $\cR_{t,T}$ is a vector space, $E(p_xX) = 0$ for all $X\in\cR_{t,T}$. As $P(p_x>0) >0$, we may renormalize and suppose that $||p_x||_{\infty} = 1$.
Let us consider the family $G = (\Gamma_x)_{x\in I}$ where $I = L^1(\R,\cF_T)\cap(\cA_{t,T}\setminus\cA_{t,T}^0)$ and $\Gamma_x = \{p_x >0\}$. For any $\Gamma\in\cF_T$ such that $P(\Gamma) > 0$, $x = 1_{\Gamma}\in I$ since $\cA_{t,T}\cap L^0(\R_-,\cF_T) = \{0\}$. Therefore, $E(xp_x)=E(1_{\Gamma}p_x) > 0$ implies that $P(\Gamma_x\cap\Gamma) > 0$. By Lemma 2.1.3 in \cite{KS}, we deduce a countable family $(x_i)_{i=1}^\infty$ of $I$ such that $\Omega = \bigcup_{i=1}^\infty\Gamma_{x_i}$. We define $p = \sum_{i=1}^\infty 2^{-i}p_{x_i}$. We have $p >0 $ a.s and we  renormalize $p$ such that $p\in L^\infty(\R_+,\cF_T)$ and $E_P(p) = 1$. We define $Q\sim P$ such that $dQ/dP = p$. We have $E(pX) = 0$ for all $X\in\cR_{t,T}$. Therefore, with $F_{u-1}\in\cF_{u-1}$, $1_{F_{u-1}}\Delta S_u\in\cR_{t,T}$ if $u\ge t+1$, so $E_Q(1_{F_{u-1}}\Delta S_u) = 0$. This implies that $E_Q(\Delta S_u|\cF_{u-1}) = 0$, i.e $(S_u)_{u=t}^T$ is a $Q$-martingale. \smallskip

Moreover, by the the construction of $Q$ above, for all $x\in \cA_{t,T}\cap L^1(\R,\cF_T)$, we have $E_Q(x|\cF_t)\ge 0$. By truncature and homogeneity, we may extend this property to every $x$ such that $E(|x||\cF_t)<\infty$ a.s. since $x/(1+E(|x||\cF_t))$ is integrable. Finally, this also holds if $E_Q(x^-|\cF_t)<\infty$ a.s.. At last, since $\rho_t(X)+X\in \cA_{t,T}$, we may conclude that $\rho_t(X)\ge -E_Q(X|\cF_t)$,  for all  $X$ such that $E_Q(X^-|\cF_t)<\infty$ a.s.. If  $x\in \cA_{t,T}\setminus\cA_{t,T}^0$, it suffices to consider the probability measure $Q_{x}=\frac{1}{2}(Q+\tilde Q)$ where $\tilde Q$ is defined by its density $d\tilde Q/d\bP=p_x$. Indeed, since $E_{\tilde Q}(x)> 0$ and $E_Q(x)\ge 0$, this implies that $E_{ Q_x}(x)> 0$  hence $\bP(E_{Q_{x}}(x|\cF_t)\ne 0)>0$.
\smallskip

 Assume that 7) holds. For some martingale measure $Q\sim P$ we have $\rho_t(\theta_t\Delta S_{t+1}) \geq -E_Q(\theta_t\Delta S_{t+1}|\cF_{t}) = 0$, hence AIP holds. If $\rho_t(\theta_t\Delta S_{t+1}) = 0$ on some non null set $\Lambda_t$, we have $\rho_t(\theta_t\Delta S_{t+1}1_{\Lambda_t}) = 0$. 
This implies $\theta_t\Delta S_{t+1}1_{\Lambda_t}$ is acceptable. Moreover, if $\theta_t\Delta S_{t+1}1_{\Lambda_t}\notin\cA_{t,T}^0$,  $E_{Q_x}(\theta_t\Delta S_{t+1}1_{\Lambda_t}|\cF_t) \neq 0$ by 7), which yields contradiction . Therefore, $\rho_t(\theta_t\Delta S_{t+1})=\rho_t(-\theta_t\Delta S_{t+1}) = 0$ on $\Lambda_t$, i.e. SRN holds, and we deduce that  1)  holds. Note that 5) and 6) are equivalent by Theorem \ref{ClosureNA}.\end{proof}

 \noindent Proof of Theorem \ref{SHP}. \begin{proof}

By Theorems \ref{ClosureNA} and \ref{FTAP}, we know that $\Gamma_{0,T}$ is closed in probability. For any $h_T\in\Gamma_{0,T}$, there exists $\sum_{t=0}^{T}\theta_{t-1}\Delta S_{t}\in \cR_{0,T}$ such that $\rho_0\left(\sum_{t=0}^{T}\theta_{t-1}\Delta S_{t} - h_T\right)\leq 0$. Since, $h_T\in L_S$, we suppose w.l.o.g that $S_T$ and $h_T$ are integrable under $P$. 

Set $\gamma_t:= \sum_{t=0}^{t}\theta_{t-1}\Delta S_{t} -h_T$ for every $t\leq T$. For any $Q\in\cQ^e\neq\emptyset$, we have:
\begin{align*}
    |\gamma_T| \leq \left|\sum_{t=0}^{T-1}\theta_{t-1}\Delta S_{t}\right| + |\theta_{T-1}||\Delta S_{T}| + |h_T|,
\end{align*}
hence:
\begin{align*}
    E_Q(|\gamma_T||\cF_{T-1}) \leq \left|\sum_{t=0}^{T-1}\theta_{t-1}\Delta S_{t}\right| + |\theta_{T-1}|E_Q(|\Delta S_{T}||\cF_{T-1}) + E_Q(|h_T||\cF_{T-1}) < \infty \text{ a.s..}
\end{align*}
By Statement 7) of Theorem \ref{FTAP} and the martingale property, we deduce that:
\begin{align} \label{dyna T-1}
    \rho_{T-1}(\gamma_T) \geq -E_Q(\gamma_{T-1}|\cF_{T-1}).
\end{align}
At time $T-2$, by time-consistency of the risk measure and (\ref{dyna T-1}), we get that
\begin{align*}
    \rho_{T-2}(\gamma_T) = \rho_{T-2}(-\rho_{T-1}(\gamma_T)) \geq \rho_{T-2}(E_Q(\gamma_{T-1}|\cF_{T-1})).
\end{align*}
Moreover, $ E_Q(|E_Q(\gamma_{T-1}|\cF_{T-1})||\cF_{T-2}) \leq E_Q(|\gamma_{T-1}||\cF_{T-2})$ 
and
\bean
    E_Q(|\gamma_{T-1}||\cF_{T-2}) &\le& \left |\sum_{t=0}^{T-2}\theta_{t-1}\Delta S_{t}\right | + |\theta_{T-2}|E_Q(|\Delta S_{T-1}||\cF_{T-2})\\
    &&+ E_Q(|h_T||\cF_{T-2}) < \infty \text{ a.s..}
\eean
We deduce by Statement 7) of Theorem \ref{FTAP} that 
$$\rho_{T-2}(E_Q(\gamma_{T-1}|\cF_{T-1})) \geq -E_Q(\gamma_{T-1}|\cF_{T-2}).$$ By the martingale property, we finally deduce that $\rho_{T-2}(\gamma_T) \geq  -E_Q(\gamma_{T-2}|\cF_{T-2})$.
Recursively, we finally obtain:
\begin{align} \label{eq: min price}
    0\geq \rho_0\left(\sum_{t=0}^{T}\theta_{t-1}\Delta S_{t} - h_T\right) \geq -E_Q(\gamma_1|\cF_0) \ge -E_Q(\theta_0\Delta S_1 - h_T) \ge  E_Q(h_T).
\end{align}
This implies $\Gamma_{0,T}\subset\Theta_{0,T}$. 

Reciprocally, assume that there is $\hat{h}_T\in\Theta_{0,T}\setminus\Gamma_{0,T}$. Since $\hat{h}_T\in L_S(\R,\cF_T)$, $\hat{h}_T$ is integrable under $Q\in\cQ^e$. Moreover, since $\Gamma_{0,T}$ is closed in probability, $\tilde{\Gamma}_{0,T}:= \Gamma_{0,T}\cap L^1_{Q}(\R,\cF_T)$ is closed in $L^1$. By the Hahn-Banach separation theorem, as $\hat{h}_T\notin \tilde{\Gamma}_{0,T}$, we deduce the existence of $Y\in L^\infty(\R,\cF_T)$ such that:
\begin{align*}
    \sup_{X\in\tilde{\Gamma}_{0,T}}E_Q(YX) < E_Q(Y\hat{h}_T).
\end{align*}
Let $H$ be the density $Q$ w.r.t $P$, i.e. $H=dQ/dP$. We have:
\begin{align*}
    \sup_{X\in\tilde{\Gamma}_{0,T}}E(HYX) < E(HY\hat{h}_T).
\end{align*}
Since $\tilde{\Gamma}_{0,T}$ is a cone, we deduce that $E(HYX)\le 0$ for all $X\in\tilde{\Gamma}_{0,T}$. Moreover,  $E(HY\hat{h}_T) > 0$, $HY\ge 0$ a.s. and $E(HY)>0$. Therefore, we deduce that $\hat{H}:= HY/E(HY)$ defines the density of a probability measure $\hat{Q}\in\cQ^a$. 

We define $H^\epsilon:= \epsilon H + (1-\epsilon)\hat{H}$. Since $E(\hat{H}\hat{h}_T) >0$, we  may choose $\epsilon\in (0,1)$ small enough so that $E(H^\epsilon\hat{h}_T) > 0$. Since $H^\epsilon$ defines the density of a probability measure $Q^\epsilon\in\cQ^e$, we should have $E_{Q^\epsilon}\hat{h}_T=E(H^\epsilon\hat{h}_T) \le 0$, as $\hat{h}_T\in\Theta_{0,T}$. This yields a contradiction. We conclude that $\Gamma_{0,T} = \Theta_{0,T}$.\smallskip

At last, $P_0$ is a super-hedging price for $h_T$ if and only if $h_T-P_0\in \Gamma_{0,T}$. By the first part, we deduce that    $P_0^* \geq \sup_{Q\in\cQ^e}E_Q(h_T)$. Suppose there exists $\epsilon >0$ such that $P_0^* - \epsilon\geq \sup_{Q\in\cQ^e}E_Q(h_T)$. Then, $(h_T - P_0^* +\epsilon)\in\Theta_{0,T}$. Since $\Theta_{0,T} = \Gamma_{0,T}$,  there exists $W_{0,T}\in\cR_{0,T}$ such that $\rho_0(W_{0,T} - h_T + P_0^* -\epsilon) \leq 0$. This  implies that $P_0^* - \epsilon \geq\rho_0(W_{0,T}-h_T)$. Since $\rho_0(W_{0,T}-h_T)$ is a super-hedging price for $h_T$, we also deduce that $\rho_0(W_{0,T}-h_T)\geq P_0^*$ which yields a contradiction. We conclude that $P_0^* = \sup_{Q\in\cQ^e}E_Q(h_T)$. \end{proof}


\begin{thebibliography}{100}

\bibitem{Acc} Acciao B. and Penner I. Dynamics convex risk measures. Advanced Mathematical Methods for Finance, 2010, 1, Ed. Giulia Di Nunno and Bent \O ksendal,  Springer Heidelberg Dordrecht London New-York.


\bibitem{AHR}Ararat C., Hamel A. and Rudloff B. Set-valued shortfall and divergence risk measures. International Journal of Theoretical and Applied Finance,  20 (5), 2017.

\bibitem{BCL}Baptiste J., Carassus L. and L\'epinette E. Pricing without martingale measures. Preprint. %\url{https://hal.archives-ouvertes.fr/hal-01774150}.


\bibitem{BL} Ben Tahar I. and L\'epinette E. Vector valued coherent risk measure processes. IJTAF, 17, 02 (2014). 

\bibitem{BS} Black F. and Scholes M. The pricing of options and corporate liabilities. Journal of Political Economy, 81, 637-659, 1973.

\bibitem{Ch1}
Cherny A. Pricing and hedging European options with discrete-time coherent risk. Finance and Stochastics, 13, 537-569, 2007. 



\bibitem{Ch2}
Cherny A. Pricing with coherent risk. Theory of Probability and Its Applications,  52(3), 389-415, 2007.

\bibitem{CL} Carassus L. and L\'epinette E. Pricing without no-arbitrage condition in discrete-time.  Journal of Mathematical Analysis and Applications, 505, 1, 125441,2021. 

\bibitem{DMW}
Dalang E.C., Morton A. and Willinger W. Equivalent martingale measures and no-arbitrage in stochastic securities market models. Stochastics and Stochastic Reports,  29, 85-201, 1990.

\bibitem{DScha1}
Delbaen F. and Schachermayer W. A general version of the fundamental theorem
of asset pricing. Mathematische Annalen, 300, 463-520, 1994.


\bibitem{DScha2}
Delbaen F. and Schachermayer W. The fundamental theorem of asset pricing for
unbounded stochastic processes. Mathematische Annalen, 312, 215-250, 1996.

\bibitem{DScha3}
Delbaen F. and Schachermayer W. The no-arbitrage condition under a change of num\'eraire. Stochastics and Stochastic Reports, 53, 213-226, 1995.

\bibitem{DScha4}
Delbaen F. and Schachermayer W. The mathematics of arbitrage. Springer Finance, 371, 2006.



\bibitem{Del1}
Delbaen F. Coherent risk measures. Lecture Notes, Cattedra Galileiana, Scuola
Normale Superiore di Pisa, 2000.


\bibitem{Del2}
Delbaen F. Coherent risk measures on general probability spaces. Advances in Finance and Stochastics: essays in honor of Dieter Sondermann, Springer, Heidelberg, 1-37, 2002.



\bibitem{DS}
Detlefsen K. and Scandolo G. Conditional and dynamic convex risk measures. Finance and Stochastics, 9, 539-561, 2005.



\bibitem{FR} Feinstein Z. and  Rudloff B. 2013. Time consistency of dynamic risk measures in markets with transaction costs. Quantitative Finance,  13 (9), 1473-1489, 2013.

\bibitem{FP} F\" olmer H. and Penner I. Convex risk measures and the dynamics of their penalty functions. Statistics and Decisions, 2, 24, 2006.

\bibitem{GRS} Guasoni P, R\'asonyi M. and Schachermayer W. The fundamental theorem of asset pricing for continuous processes under small transaction costs. Annals of Finance, 6, 157-191, 2010.

\bibitem{HK} Harrison J.M. and Pliska S.R. Martingales and stochastic integrals in the theory of continuous trading. Stochastic Processes and their Applications, 11, 215-260, 1981.

\bibitem{hes02} Hess C. Set-valued integration and set-valued probability theory: An
  overview. In: E.~Pap (ed.) Handbook of Measure Theory,  Elsevier, 14, 617-673, 2002.
 
 \bibitem{JS} Jacod J. and Shiryaev A.N. Local martingales and the fundamental asset pricing theorems in the discrete-time case. Finance and Stochastics,  2(3), 259-273, 1998.
 
 \bibitem{Jorion} Jorion P.  Value at Risk: the new benchmark for managing financial risk. McGraw-Hill Professional, 3,  2006.
 
 \bibitem{JMT} Jouini E., Meddeb M. and  Touzi N. Vector-valued measure of risk. Finance and Stochastics, 8, 531-552, 2004.
 
 \bibitem{KK} Kabanov Y.M. and Kramkov D. No arbitrage and equivalent martingale measures: an elementary proof of the Harrison-Pliska theorem. Theory of Probability and its Applications,  39(3), 523-527, 2001.

\bibitem{KSt}
Kabanov Y. and Stricker C. A Teachers' note on no-arbitrage criteria. In
S\'eminaire de Probabilit\'es, XXXV, volume 1755 of Lecture Notes in Math., Springer Berlin, 149-152, 2001.


\bibitem{KS}
Kabanov Y. and Safarian, M. Markets with transaction costs. Mathematical Theory. Springer-Verlag, 2009.


\bibitem{KL}
Kaina M. and Ruschendorf L. On convex risk measures on $L^p$-spaces. Mathematical Methods of Operations Research, 69(3), 475-495,  2009.

\bibitem{LM} L\'epinette E. and Molchanov I. Risk arbitrage and hedging to  acceptability. Finance and Stochastics, 25, 101-132, 2021.

\bibitem{ZL}  L\'epinette E. and Zhao J. Super-hedging a European option with a coherent risk-measure and without no-arbitrage condition. Preprint. 


\bibitem{Merton} Merton R.C. The theory of rational option pricing. Bell J. Econ. Manag. Sci.,  4, 141-183, 1973.

\bibitem{Mol}
Molchanov I. Theory of Random Sets. 2nd edition.  Springer, London, 2017.





\bibitem{P} Pergamenshchikov S.  Limit theorem for Leland's strategy.  The Annals of Applied Probability, 13, 3, 1099-1118, 2003.


\bibitem{RW}
Rockafellar, R Tyrrell and Wets, Roger J-B. Variational analysis. Springer Science \& Business Media, 2009.

\bibitem{Rogers} Rogers L.C.G. Equivalent martingale measures and no-arbitrage. Stochastics and Stochastic Reports, 51, 41-49, 1994.

\bibitem{Ross} Ross S. The arbitrage theory of capital asset pricing. Journal of Economic Theory, 13, 341-360, 1976. 

\bibitem{DScha5} Schachermayer W.  A Hilbert space proof of the fundamental theorem of asset pricing in finite discrete time. Insurance: Mathematics and Economics  11(4), 249-257, 1992.



\end{thebibliography}
\end{document}